\documentclass{revtex4}
\usepackage{amsmath,amsthm,amssymb}
\usepackage{rotating}
\usepackage{hhline}
\usepackage{graphicx}

\usepackage{amsfonts}
\newtheorem{theorem}{Theorem}[section]
\newtheorem{lemma}[theorem]{Lemma}
\newtheorem{proposition}[theorem]{Proposition}

\theoremstyle{remark}
\newtheorem{remark}[theorem]{Remark}

\theoremstyle{definition}
\newtheorem{definition}[theorem]{Definition}

\theoremstyle{example}
\newtheorem{example}[theorem]{Example}

\theoremstyle{notation}

\newcommand{\bra}[1]{\langle#1|}
\newcommand{\ket}[1]{|#1\rangle}
\usepackage{xcolor}
\begin{document}

\title{ The Heisenberg-Weyl-parity group its coherent states and a unified Wigner-Weyl function}            
\author{A. Vourdas}
\affiliation{Department of Computer Science,\\
University of Bradford, \\
Bradford BD7 1DP, United Kingdom\\a.vourdas@bradford.ac.uk}

\begin{abstract}
The Heisenberg-Weyl group $HW(d)$ related to a $d$-dimensional Hilbert space $H(d)$, is enlarged into the Heisenberg-Weyl-parity group $HWP(d)$ that incorporates parity transformations.
It consists of $2d^3$ elements, of which $d^3$ elements belong to the $HW(d)$ subgroup, and extra $d^3$ elements  which are related through a Fourier transform with the former ones. 
It is shown that $HWP(d)$ is a generalised version of the dihedral group. The properties of operators that combine displacements and parity, are discussed. 
$HWP(d)$ is shown to be a solvable group, and commutators of its elements  perform displacement and parity transformations of quantum states, along loops in the discrete phase space.
$2d^2$ coherent states related to the $HWP(d)$ group are introduced, which consist of $d^2$ coherent states related to the $HW(d)$ subgroup, and extra $d^2$ coherent states which are 
related through a Fourier transform with the former ones. In noisy cases, expansion of an arbitrary state in terms of the $2d^2$ coherent states with Bargmann coefficients, is advantageous in comparison to expansion in terms of the $d^2$ coherent states related to $HW(d)$. One of the consequences of the $HWP(d)$ group, is a natural unification of the Wigner and Weyl functions. 
The properties of the unified Wigner-Weyl function are discussed.

\end{abstract}
\maketitle

\section{Introduction}
Phase space methods for systems with infinite-dimensional Hilbert space have been studied extensively in the literature (e.g. \cite{S, Z}). 
Applications to quantum tomography are reviewed in\cite {TT}.
Related work in the general area of Harmonic analysis is reviewed in \cite{10}, and in wavelets for  signal analysis in \cite{signal}.

In the present paper we are interested in
phase space methods for systems with finite ($d$-dimensional) Hilbert space $H(d)$. They  have been studied originally by Weyl and Schwinger and later by many authors 
(various aspects of this area are reviewed in Refs\cite{1,2,2A} and in references therein).
They involve topics like finite Fourier transforms and the Heisenberg-Weyl group $HW(d)$ (\cite{3,4}), Wigner and Weyl functions (\cite{5,6,6AA,6A,6B,6C,6D,6E,6F,6G}), 
symplectic transformations \cite{ST1, ST2, ST3}, tomograms \cite{MA1,MA2, MA3}, coherent states\cite{7,8,9,9A,9B}, etc.
An important tool for Wigner functions (in both finite and infinite-dimensional Hilbert space) and more generally for phase space methods, are the parity transformations (\cite{11A,11,12,13} ).

In this general context of quantum systems with finite-dimensional Hilbert space, in the present paper we discuss the following: 
\begin{itemize}
\item
We enlarge the Heisenberg-Weyl group $HW(d)$, into the `Heisenberg-Weyl-parity' group $HWP(d)$ 
that incorporates parity transformations. 
The displacement operators (Eq.(\ref{disA})) and the displaced parity operators (Eq.(\ref{201})) are unified within the $HWP(d)$ group. 
The Heisenberg-Weyl group $HW(d)$ is a subgroup of $HWP(d)$, and we show that the extra elements in $HWP(d)$ (the displaced parity operators) 
are related through a Fourier transform with the displacement operators which are the elements of $HW(d)$.
A mathematical tool that brings extra symmetries to an existing group is the semidirect product. In the present context the 
$HWP(d)$ is the semi-direct product of the Heisenberg-Weyl group $HW(d)$ times ${\mathbb Z}(2)$ which is related to parity.

\item 
$HWP(d)$ is shown to be solvable group, 
and the physical importance of its solvability is discussed. Solvable groups\cite{BOU} are `mildly non-Abelian', and their study uses commutators which in the present context
perform displacements of quantum states along loops in the discrete phase space. 
For a solvable group with solvability class $n$,  there are $(n-1)$ levels of non-commutativity, each of which is described with loops related to commutators.
The solvability class of $HWP(d)$ is $n=3$.

\item
Acting with the $2d^2$ elements of the $HWP(d)$ group on a fiducial vector, we get $2d^2$ coherent states related to the $HWP(d)$ group. They consist of $d^2$ coherent states related to the $HW(d)$ subgroup, and extra $d^2$ coherent states which are 
related through a Fourier transform with the former ones. Using the resolution of the identity, we expand an arbitrary state in terms of coherent states with Bargmann coefficients.
\item
Two expansions of an arbitrary state are studied, in terms of the $2d^2$ coherent states related to the $HWP(d)$ group, and also in terms of the $d^2$ coherent states related to the $HW(d)$ group.
In noisy problems, the former expansion is more immune to noise, than the latter expansion.
This is because the redundancy (the number of coherent states is larger than the dimension of the space) is larger with $2d^2$ coherent states than with $d^2$ coherent states.
\item
Wigner and Weyl functions are intimately related to displaced parity operators and displacement operators, correspondingly.
The unification of the displacement operators and the displaced parity operators into the $HWP(d)$ group, leads naturally into the unification of Wigner and Weyl functions. 
The properties of the unified Wigner-Weyl function are discussed.

\end{itemize}

We next describe the logical chain that leads to these results. Sections 2-4 present preparatory material, and sections 5-7 the new results, as follows.

In section 2, we introduce briefly quantum systems with $d$-dimensional Hilbert space, in order to define the notation.
We then discuss the solvability of the Heisenberg-Weyl group $HW(d)$ (proposition \ref{pro3}), and its physical importance.
We note that the Heisenberg-Weyl group $HW(d)$ is in fact nilpotent (which is stronger concept than solvability).
Also the solvability of $HW(d)$ with odd $d$ (the case studied here), follows immediately from the Feit-Thompson theorem. 
However, our discussion of the solvability of $HW(d)$, is a first step towards the solvability of the $HWP(d)$ group (discussed in section 5).
We also introduce $d^2$ coherent states related to the $d^2$ elements of the $HW(d)$ group, 
expansions of an arbitrary state in terms of them with Bargmann coefficients, and the $Q$ (or Husimi) function in this context (proposition \ref{pro32}).

In section 3, we introduce parity operators and discuss briefly the physical properties of the displaced parity operators.
We also define Wigner and Weyl functions in Eq.(\ref{WIG}), discuss their marginal properties and their relationship through a Fourier transform.

In section 4, we consider the dihedral group and its representation in terms of parity transformations and displacements in the momentum direction only (or in the position direction only).
This is preparatory material for the $HWP(d)$ group, that combines parity transformations with displacements in both position and momentum.

In section 5, we discuss the `Heisenberg-Weyl-parity' group $HWP(d)$ (proposition \ref{pro123}), the properties of its elements (proposition \ref{pro1}), and 
the physical importance of its solvability (proposition \ref{pro45}).
We also introduce $2d^2$ coherent states related to the $2d^2$ elements of the $HWP(d)$ group (proposition \ref{pro46}), 
expansions of an arbitrary state in terms of them, and the $Q$ (or Husimi) function in this context. 

In section 6, coherent states are used in noisy problems. Due to redundancy (the number of coherent states is greater than the dimension of the Hilbert space) 
expansions in terms of coherent states are immune to noise. Examples show that the use  of $2d^2$ coherent states related to the $HWP(d)$ group
is more immune to noise,  than the use  of $d^2$ coherent states related to the $HW(d)$ group (because in the former case the redundancy is greater).

In section 7,  using the elements of the $HWP(d)$ group, we get naturally a unified Wigner-Weyl function. We study its properties (proposition \ref{pro7}) and present numerical results for some examples.

We conclude in section 8 with a discussion of our results and comments about further work.

\section{The Heisenberg-Weyl group as a solvable group}
\subsection{Quantum systems with variables in ${\mathbb Z}(d)$ with odd $d$}\label{sec19}

We  consider a quantum system with variables in  the ring ${\mathbb Z}(d)$ of integers modulo $d$. 
$H(d)$ is the $d$-dimensional Hilbert space describing this system. Such systems have been studied extensively in the literature for a long time (for a review see \cite{1,2,2A}).
There are well known technical differences related to number theory, between quantum systems with odd dimension $d$ and even dimension $d$\cite{WQ1}.
Unifying approaches for these two cases and for the continuous variables case, have been discussed in \cite{WQ2,WQ3,WQ4}.
In this paper we consider systems with odd dimension $d$, and then
the inverse of $2$ exists as an integer in ${\mathbb Z}(d)$ ($2^{-1}=\frac{d+1}{2}$). Its existence is very important in many relations below.

Let $|X;j\rangle$ where $j\in {\mathbb Z}(d)$ be an orthonormal basis in $H(d)$. $X$ in the notation is not a variable, it simply indicates `position states'.
The finite Fourier transform $F$ is given by
\begin{eqnarray}\label{FF}
&&F=\frac{1}{\sqrt{d}}\sum _{j,k}\omega(jk) \ket{X;j}\bra{X;k};\;\;\;\omega(\alpha)=\exp \left (i\frac{2\pi\alpha}{d}\right)\;\;\;\alpha,j,k\in{\mathbb Z}(d)\nonumber\\
&&F^4={\bf 1};\;\;\;FF^{\dagger}={\bf 1};\;\;\;\frac{1}{d}\sum _a\omega(\alpha)=\delta_{\alpha,0}.
\end{eqnarray}
We act with $F$ on position states and get the basis of `momentum states':
\begin{eqnarray}
&&\ket{P;j}=F\ket{X;j}.
\end{eqnarray}
$P$ in the notation is not a variable, it simply indicates `momentum states'.

We can define position and momentum operators as
\begin{eqnarray}\label{GG}
{\widehat x}=\sum j\ket{X;j}\bra{X;j};\;\;\;{\widehat p}=\sum j\ket{P;j}\bra{P;j};\;\;\;j=-\frac{d-1}{2},...,\frac{d-1}{2}.
\end{eqnarray}
In the present context position and momentum are angular quantities,  and like all angular quantities ${\widehat x}$, ${\widehat p}$ are multivalued.
But exponentials of angular quantities times $i$ are single-valued, and below we introduce and work mainly with such single-valued quantities. 

We introduce the displacement operators 
\begin{eqnarray}\label{A}
&&X^\beta = \omega(-\beta {\widehat p})=\sum _j\omega (-j\beta)|P; j \rangle \bra{P; j }=\sum _j\ket{X; j+\beta}\bra{X;j};\nonumber\\
&&Z^\alpha =\omega(\alpha {\widehat x})=\sum _j|P; j+\alpha \rangle \bra{P;j}= \sum _j\omega (\alpha j)|X; j\rangle\bra{X;j}=F X^\alpha F^\dagger;\nonumber\\
&&X^{d}=Z^{d}={\bf 1};\;\;\;X^\beta Z^\alpha= Z^\alpha X^\beta \omega(-\alpha \beta);\;\;\;\alpha, \beta\in {\mathbb Z}(d).
\end{eqnarray}
$Z,X$ are displacement operators in the (discrete) momentum direction and (discrete) position direction, respectively.

General displacement operators are the unitary operators
\begin{eqnarray}\label{disA}
D(\alpha, \beta, \gamma)=Z^\alpha X^\beta \omega (\gamma-2^{-1}\alpha \beta);\;\;\;[D(\alpha, \beta, \gamma)]^\dagger=D(-\alpha, -\beta, -\gamma);\;\;\;[D(\alpha, \beta, \gamma)]^d={\bf 1}.
\end{eqnarray}
The $D(\alpha, \beta, \gamma)$ form a representation of the 
Heisenberg-Weyl group $HW(d)$, which has $d^3$ elements, with multiplication:
\begin{eqnarray}\label{199F}
D(\alpha_1, \beta_1,\gamma_1)D(\alpha_2, \beta_2,\gamma_2)=D[\alpha_1+\alpha_2, \beta_1+\beta_2,\gamma_1+\gamma_2+2^{-1}(\alpha_1\beta_2-\alpha_2\beta_1)].
\end{eqnarray}
It is easily seen that
\begin{eqnarray}\label{122}
&&D(\alpha, \beta, \gamma)\ket{X;j}=\omega(2^{-1}\alpha\beta+\alpha j+\gamma)\ket{X;j+\beta}\nonumber\\
&&D(\alpha, \beta, \gamma)\ket{P;j}=\omega(-2^{-1}\alpha\beta-\beta j+\gamma)\ket{P;j+\alpha},
\end{eqnarray}
and that
\begin{eqnarray}\label{1200}
D(\alpha, \beta, \gamma){\widehat x}[D(\alpha, \beta, \gamma)]^\dagger={\widehat x}-\beta {\bf 1};\;\;\;
D(\alpha, \beta, \gamma){\widehat p}[D(\alpha, \beta, \gamma)]^\dagger={\widehat p}-\alpha {\bf 1}.
\end{eqnarray}
Physically we can implement stroboscopically displacement transformations, with a Hamiltonian that is the principal logarithm of $D(\alpha, \beta, \gamma)$:
\begin{eqnarray}\label{HAM}
h=\frac{d}{2\pi i}\log D(\alpha, \beta, \gamma)=\frac{d}{2\pi i}\log D(\alpha, \beta, 0)+\gamma {\bf 1}.
\end{eqnarray}
Then stroboscopically we get the time evolution operator
\begin{eqnarray}\label{10A}
\exp(iht)=D(\alpha, \beta, \gamma);\;\;\;t=\frac{2\pi }{d}+2\pi N;\;\;\;N=0,1,...
\end{eqnarray}

For later use we introduce the following three subgroups of $HW(d)$:
\begin{eqnarray}\label{sub}
&&{\cal G}_d(Z)=\{Z^a|a\in {\mathbb Z}(d)\}\cong {\mathbb Z}(d)\nonumber\\
&&{\cal G}_d(X)=\{X^a|a\in {\mathbb Z}(d)\}\cong {\mathbb Z}(d)\nonumber\\
&&{\cal G}_d({\bf 1})=\{\omega(a){\bf 1}|a\in {\mathbb Z}(d)\}\cong {\mathbb Z}(d).
\end{eqnarray}
They are all isomorphic to ${\mathbb Z}(d)$.

\subsection{The physical importance of the solvability of a group}

The commutator of two elements $g,h$ of a group $G$ is the element $[g,h]=ghg^{-1}h^{-1}$.
If $g$ and $h$ commute, then $[g,h]={\bf 1}$. 
We can prove that
 \begin{eqnarray}\label{TYU}
&&[g,h]^{-1}=[h,g];\;\;\;[g,h]hg=gh\nonumber\\
&&[f,gh][g,hf][h,fg]={\bf 1}.
\end{eqnarray}
All the commutators of elements of $G$ generate a normal subgroup of $G$, which is called commutator or derived group of $G$ and is denoted as $[G,G]$.
The commutator group $[G,G]$ is a measure of the non-Abelian nature of $G$ (for an Abelian group $[G,G]=\{{\bf 1}\}$).

 We define the following normal subgroups of $G$ ($\triangleright$ indicates normal subgroup):
 \begin{eqnarray}
G;\;\;\;G_1=[G,G];\;\;\;G_2=[G_1,G_1];\;\;\;G_3=[G_2,G_2]....
\end{eqnarray}
The derived series of $G$ is
\begin{eqnarray}
G\triangleright G_1\triangleright  G_2\triangleright G_3...
\end{eqnarray}
If for some $n$ the derived series terminates in $G_n=\{{\bf 1}\}$, then the group is called solvable.
The number $n$ is called solvability class. 

For any group $G$ the $G/[G,G]$ is Abelian, and therefore the following groups are Abelian
 \begin{eqnarray}\label{100}
G/G_1;\;\;\;G_1/G_2;\;\;\;G_2/G_3....
\end{eqnarray}
They can be regarded as generalised phase-spaces related to various variables (in our case position, momentum and parity).
The commutators ${\cal L}(g,h)=[g,h]=ghg^{-1}h^{-1}$ are elements of the normal subgroup $G_1=[G,G]$ of $G$, and describe group action (in our case displacements and parity transformations) 
on quantum states, along loops (defined by $g,h$) in these phase spaces. They play an important role in the present paper. Then:
\begin{itemize}
\item
If the solvability class of $G$ is $n=1$, then $G$ is Abelian.

\item
If the solvability class of $G$ is $n=2$, then $G$ is non-Abelian. 
We will see below that this is the case with the $HW(d)$ group (that involves quantum non-commutativity between position and momentum), and also with the 
dihedral groups $\Delta_d(Z)$, $\Delta_d(X)$ (that involve non-commutativity between parity and one of the displacement operators $X,Z$). In this case 
the commutators ${\cal L}(g,h)=[g,h]$ are elements of the Abelian group $G_1=[G,G]$. 
\item
If the solvability class of $G$ is $n=3$,  there are two different levels of non-commutativity. 
We will see below that the $HWP(d)$ group that we introduce in this paper, has both the quantum non-commutativity between position and momentum,
and also the non-commutativity between parity and the displacement operators.
In this case $G_1=[G,G]$ is a non-Abelian group, and $G_2=[G_1,G_1]$ is an Abelian group.
Then the commutators ${\cal L}^{(1)}(g,h)=[g,h]$ are elements of the non-Abelian group $G_1=[G,G]$, and for this reason
we also introduce the `commutators of the commutators' 
\begin{eqnarray}\label{160}
{\cal L}^{(2)}(g_1,h_1|g_2,h_2)&=&[{\cal L}^{(1)}(g_1,h_1),{\cal L}^{(1)}(g_2,h_2)]={\cal L}^{(1)}(g_1,h_1){\cal L}^{(1)}(g_2,h_2)[{\cal L}^{(1)}(g_1,h_1)]^{-1}[{\cal L}^{(1)}(g_2,h_2)]^{-1}\nonumber\\
&=&{\cal L}^{(1)}(g_1,h_1){\cal L}^{(1)}(g_2,h_2){\cal L}^{(1)}(h_1,g_1){\cal L}^{(1)}(h_2,g_2)\nonumber\\
&=&(g_1h_1g_1^{-1}h_1^{-1})(g_2h_2g_2^{-1}h_2^{-1})(h_1g_1h_1^{-1}g_1^{-1})(h_2g_2h_2^{-1}g_2^{-1}),
\end{eqnarray}
which are elements of the Abelian group $G_2=[G_1,G_1]$.
\item
We do not have groups with solvability class $n\ge 4$ in this paper, but in this case there are $(n-1)$ levels of non-commutativity, and we can introduce commutators
${\cal L}^{(1)},...,{\cal L}^{(n-1)}$.

\end{itemize}

In the next subsection we give details of this for the $HW(d)$ group, and later for the $\Delta  _d(Z)$ and $HWP(d)$ groups.

A stronger concept than solvability, is nilpotency.
 For nilpotency, we check that the following lower central series of a group $G$ terminates in $\{{\bf 1}\}$:
 \begin{eqnarray}
 G;\;\;\;[G,G];\;\;\;[G,[G,G]];\;\;\;[G,[G,[G,G]]]; ...
  \end{eqnarray} 
We will see below that $HW(d)$ is both nilpotent and solvable group, but $HWP(d)$ which is the main theme of this paper is solvable but not nilpotent.
For this reason we discuss solvability but not nilpotency.

\subsection{$HW(d)$ as solvable group and commutators that perform displacements along loops in phase space}\label{sec1}
\begin{lemma}
\mbox{}
\begin{itemize}
\item[(1)]
The commutator of two elements of the $HW(d)$ group is 
\begin{eqnarray}\label{199n}
&&{\cal L}_1(\alpha_1, \beta_1|\alpha_2, \beta_2)=[D(\alpha_1, \beta_1,\gamma_1),D(\alpha_2, \beta_2,\gamma_2)]
=\omega[{\cal A}(\alpha_1,\beta_1|\alpha_2,\beta_2)]{\bf 1}\nonumber\\
&&{\cal A}(\alpha_1,\beta_1|\alpha_2,\beta_2)=\alpha_1\beta_2-\alpha_2\beta_1\in {\mathbb Z}(d)
\end{eqnarray}
It does not depend on $\gamma_1, \gamma_2$. 
\item[(2)]
\begin{eqnarray}\label{200}
[D(\alpha_1, \beta_1,\gamma_1)D(A_1, B_1,\Gamma_1),D(\alpha_2, \beta_2,\gamma_2)]&=&[D(\alpha_1, \beta_1,\gamma_1),D(\alpha_2, \beta_2,\gamma_2)]\nonumber\\&\times&
[D(A_1, B_1,\Gamma_1),D(\alpha_2, \beta_2,\gamma_2)]
\end{eqnarray}
Consequently
\begin{eqnarray}\label{199}
{\cal A}(\alpha_1+A_1, \beta_1+B_1|\alpha_2, \beta_2)={\cal A}(\alpha_1, \beta_1|\alpha_2, \beta_2)+{\cal A}(A_1, B_1|\alpha_2, \beta_2).
\end{eqnarray}
\end{itemize}
\end{lemma}
\begin{proof}
\mbox{}
\begin{itemize}
\item[(1)]
The commutator of two elements of the Heisenberg-Weyl group is 
\begin{eqnarray}
{\cal L}_1(\alpha_1, \beta_1|\alpha_2, \beta_2)&=&D(\alpha_1, \beta_1,\gamma_1)D(\alpha_2, \beta_2,\gamma_2)[D(\alpha_1, \beta_1, \gamma_1)]^\dagger[D(\alpha_2, \beta_2, \gamma_2)]^\dagger\nonumber\\
&=&D(0,0,\alpha_1\beta_2-\alpha_2\beta_1)=
\omega(\alpha_1\beta_2-\alpha_2\beta_1){\bf 1}.
\end{eqnarray}

\item[(2)]
Eq.(\ref{200}) is proved using Eq.(\ref{199F}) together with Eq.(\ref{199n}).
From this follows Eq.(\ref{199}).
\end{itemize}
\end{proof}
\begin{proposition}\label{pro3}
\mbox{}
\begin{itemize}
\item[(1)]

\begin{eqnarray}\label{13}
[HW(d),HW(d)]={\cal G}_d({\bf 1});\;\;\;[{\cal G}_d({\bf 1}),{\cal G}_d({\bf 1})]\cong\{{\bf 1}\}.
\end{eqnarray}
\item[(2)]
$HW(d)$ is a solvable group (of solvability class $2$) and  its derived series is
\begin{eqnarray}\label{101}
HW(d)\triangleright {\cal G}_d({\bf 1})\triangleright \{{\bf 1}\}.
\end{eqnarray}
Physically the ${\cal G}_d({\bf 1})$ is related here, with the quantum non-commutativity between $Z$ and $X$.
\item[(3)]
The Abelian groups in Eq.(\ref{100}) are here
\begin{eqnarray}\label{150}
HW(d)/{\cal G}_d({\bf 1})\cong {\mathbb Z}(d)\times {\mathbb Z}( d);\;\;\;{\cal G}_d({\bf 1})/\{{\bf 1}\}={\cal G}_d({\bf 1})\cong {\mathbb Z}(d).
\end{eqnarray}
Physically ${\mathbb Z}(d)\times {\mathbb Z}( d)$ is the position-momentum phase space of the quantum system (discretised torus) which we will denote as $PS(X,Z)$ (phase space for $X,Z$).
\item[(4)]
$HW(d)$ is also nilpotent group. 
\end{itemize}
\end{proposition}

\begin{proof}
\mbox{}
\begin{itemize}
\item[(1)]
 The commutators
${\cal L}_1(\alpha_1, \beta_1|\alpha_2, \beta_2)$ in Eq.(\ref{199n}) generate ${\cal G}_d({\bf 1})$, which is the commutator subgroup (or derived subgroup) of $HW(d)$. This proves the first of Eqs(\ref{13}).
${\cal G}_d({\bf 1})$ is an Abelian group, and therefore the commutator $[{\cal G}_d({\bf 1}),{\cal G}_d({\bf 1})] \cong\{{\bf 1}\}$.
\item[(2)]
From Eq.(\ref{13}), follows immediately that $HW(d)$ is a solvable group and  that its derived series is the one in Eq.(\ref{101}).
\item[(3)]
With regard to the Abelian groups in Eq.(\ref{100}), the first one is 
\begin{eqnarray}
HW(d)/[HW(d),HW(d)]\cong HW(d)/{\cal G}_d({\bf 1})\cong {\mathbb Z}(d)\times {\mathbb Z}( d)=PS(X,Z).
\end{eqnarray}
$HW(d)/{\cal G}_d({\bf 1})$ consists of cosets of displacement operators
$D(\alpha, \beta, \gamma)$ defined modulo a phase factor $\omega(\gamma)$.
$HW(d)/{\cal G}_d({\bf 1})$ is isomorphic to the phase space group $PS(X,Z)={\mathbb Z}(d)\times {\mathbb Z}( d)$ which is an additive Abelian  group, with componentwise addition 
\begin{eqnarray}\label{140} 
(\alpha, \beta)+(\gamma, \delta)=(\alpha+\gamma, \beta+\delta);\;\;\;\alpha, \beta, \gamma,\delta \in {\mathbb Z}(d).
\end{eqnarray}

The other relation ${\cal G}_d({\bf 1})/\{{\bf 1}\}={\cal G}_d({\bf 1})\cong {\mathbb Z}(d)$ is straightforward.
\item[(4)]

The lower central series of $HW(d)$ terminates in $\{\bf 1\}$:
\begin{eqnarray}
[HW(d),HW(d)]={\cal G}_d({\bf 1});\;\;\;[HW(d), [HW(d),HW(d)]]=[HW(d),{\cal G}_d({\bf 1})]\cong\{{\bf 1}\}.
\end{eqnarray}
Therefore $HW(d)$ is nilpotent group.

\end{itemize}

\end{proof}

The physical importance of solvability,  lies in the fact that the commutators $[g,h]=ghg^{-1}h^{-1}$ (which are elements of $[HW(d),HW(d)]={\cal G}_d({\bf 1})$) perform displacements of quantum states along loops in
the $PS(X,Z)$ phase space.
The ${\cal A}(\alpha_1,\beta_1|\alpha_2,\beta_2)$ in Eq.(\ref{199n}), is  area related to the magnitude of the exterior product of the vectors $(\alpha_1,\beta_1)$ and $(\alpha_2,\beta_2)$, that define the loop.
Acting with ${\cal L}_1(\alpha_1, \beta_1|\alpha_2, \beta_2) $ on a state $\ket{f}$, 
displaces this state along a loop  in the $PS( X,Z)$ phase space, and it will change its phase (but not its direction)
\begin{eqnarray}\label{14}
\ket{f}\rightarrow {\cal L}_1(\alpha_1, \beta_1|\alpha_2, \beta_2)\ket{f}=\omega[{\cal A}(\alpha_1,\beta_1|\alpha_2,\beta_2)]\ket{f}.
\end{eqnarray}
Then
\begin{eqnarray}\label{loop}
|\bra{f}{\cal L}_1(\alpha_1, \beta_1|\alpha_2, \beta_2)\ket{f}|=1.
\end{eqnarray}
This should be compared with Eqs.(\ref{341}),(\ref{loop1}) below.

We note that the phase ${\cal A}(\alpha_1, \beta_1|\alpha_2, \beta_2)$ might be only of dynamical nature, if the displacements are due to a time evolution with appropriate Hamiltonian.
But it might also include a Berry (geometric) phase 
if the system evolves with a Hamiltonian that depends on parameters that change slowly (adiabatically) around a closed loop in parameter space.

\begin{remark}\label{rem1}
The methodology used in this paper, is that of solvable groups and semi-direct products in a physical context.
Another approach to $HW(d)$ is to study central extensions of the Abelian group $HW(d)/{\mathbb Z}(d)\cong  {\mathbb Z}(d)\times {\mathbb Z}(d)$ by ${\mathbb Z}(d)$, and the two-cohomology group
$H^2[HW(d)/{\mathbb Z}(d),{\mathbb Z}(d)]$(e.g., \cite{VOUR}). In further work, this approach could be generalised to 
the more complex group $HWP(d)$ where apart from the quantum non-commutativity between position and momentum, we also have non-commutativity between parity and position and momentum.

\end{remark}

\begin{remark}
For odd $d$ the fact that $HW(d)$ is solvable, follows immediately from the general result that all finite groups of odd order are solvable (Feit-Thompson theorem\cite{FT}).
But we discuss explicitly the solvability of $HW(d)$ (in proposition \ref{pro3}), because it is the first step towards the solvability of the $HWP(d)$ group (in proposition \ref{pro45}), and the latter group and its 
physical consequences  are the main topic of the present paper. We will see that $HWP(d)$ is solvable but not nilpotent, and that it is of even order (so the Feit-Thompson theorem is not applicable to it).
\end{remark}

 \subsection{Coherent states related to the $HW(d)$ group}\label{sec24}

Coherent states\cite{7,8,9,9A} related to the $HW(d)$ group are the following $d^2$ states in $H(d)$:
\begin{eqnarray}\label{coh1}
\ket{C; \alpha, \beta}=D(\alpha, \beta,0)\ket{s};\;\;\;\ket{s}=\sum _m s_m\ket{X;r};\;\;\;\sum _m|s_m|^2=1.
\end{eqnarray}
The $C$ in the notation $\ket{C; \alpha, \beta}$ is not a variable, but it indicates coherent states.
$\ket{s}$ is a fiducial vector (reference vector), and it should be {\bf generic} so that the coherent states are different from each other 
(it should not be a position or momentum state, because in this case many coherent states differ only by a phase factor and represent the same state).
We denote as ${\cal C}[HW(d)]$ the set of these $d^2$ states.

Some properties of the coherent states are given in the following proposition.
\begin{proposition}\label{pro32}
\mbox{}
\begin{itemize}
\item[(1)]
Closure under displacement transformations in $HW(d)$: acting with displacement transformations on coherent states, we get other coherent states.

\item[(2)]
The resolution of the identity
\begin{eqnarray}\label{TR7}
\frac {1}{d}\sum _{\alpha, \beta } \ket{C; \alpha, \beta} \bra{C;\alpha, \beta}={\bf 1}.
\end{eqnarray}

\item[(3)]
An arbitrary state $\ket{f}$ in $H(d)$, can be expanded in terms of the $d^2$ coherent states:
\begin{eqnarray}\label{A12}
\ket{f}=\frac{1}{d}\sum _{\alpha, \beta } f(\alpha, \beta)\ket{C; \alpha, \beta} ;\;\;\;f(\alpha, \beta)=\bra{C;\alpha, \beta} f \rangle.
\end{eqnarray}
The $f(\alpha, \beta)$ are Bargmann coefficients (but there is no analyticity in the present discrete context ).
The scalar product of two states $\ket{f}, \ket{g}$ is given in terms of $f(\alpha, \beta)$ and $g(\alpha, \beta)$ as
\begin{eqnarray}\label{B12}
\langle g\ket{f}=\frac{1}{d}\sum _{\alpha, \beta }[g(\alpha, \beta)]^* f(\alpha, \beta).
\end{eqnarray}

\item[(4)]
The
\begin{eqnarray}\label{A12C}
Q(\alpha, \beta)=|f(\alpha, \beta)|^2;\;\;\;\frac{1}{d}\sum _{\alpha, \beta } Q(\alpha, \beta)=1,
\end{eqnarray}
is the $Q$-function or Husimi function in the present context (and depends on the fiducial vector $\ket{s}$). It takes non-negative values and the $\frac{1}{d}Q(\alpha, \beta)$
can be interpreted as a pseudo-probability distribution.

\item[(5)]
Non-orthogonality:
\begin{eqnarray}
&&|\langle C;\gamma, \delta\ket{C;\alpha, \beta}|^2=\left |\sum _m \omega [(\alpha-\gamma)m]s_{m+\beta-\delta}^*s_m\right|^2.
\end{eqnarray}
\end{itemize}
\end{proposition}

\begin{proof}
\begin{itemize}
\item[(1)]
Using Eq.(\ref{199F}) we prove the closure property:
\begin{eqnarray}
D(\alpha_1, \beta_1,\gamma_1)\ket{C; \alpha, \beta}=\omega[\gamma+2^{-1}(\alpha_1\beta-\alpha \beta_1)]\ket{C; \alpha+\alpha_1, \beta+\beta_1}.
\end{eqnarray}
\item[(2)]
The resolution of the identity in Eq.(\ref{TR7}) is a special case of the following more general relation, which holds for any operator $\Theta$ with ${\rm Tr}(\Theta)\ne 0$:
\begin{eqnarray}\label{TR}
\frac {1}{d}\sum _{\alpha, \beta } D(\alpha, \beta,0) \frac{\Theta}{{\rm Tr}(\Theta)} [D(\alpha, \beta,0)]^\dagger={\bf 1}.
\end{eqnarray}
In order to prove it, we express $\Theta$ as $\Theta=\sum \Theta_{mn}\ket{X;m}\bra{X;n}$ and then
\begin{eqnarray}
D(\alpha, \beta ,0) \Theta [D(\alpha, \beta ,0)]^\dagger=\sum_{m,n}\omega [\alpha(m-n)]\Theta_{mn}\ket{X;m+\beta}\bra{X;n+\beta}.
\end{eqnarray}
Taking into account Eq.(\ref{FF}) we prove that
\begin{eqnarray}
\frac {1}{d}\sum _{\alpha, \beta }D(\alpha, \beta ,0) \Theta [D(\alpha, \beta ,0)]^\dagger={\rm Tr}(\Theta){\bf 1}.
\end{eqnarray}
\item[(3)]
Using the resolution of the identity we prove both Eqs(\ref{A12}),(\ref{B12}).
\item[(4)]
The first part of Eq.(\ref{A12C}) is definition, and the second part follows from Eq.(\ref{B12}).
\item[(5)]
This is proved using Eq.(\ref{199F}).
\end{itemize}
\end{proof}

There is a lot of work on finding fiducial vectors (for various dimensions $d$) such that all the $|\langle C;\gamma, \delta\ket{C;\alpha, \beta}|^2$ 
with $(\gamma, \delta)\ne(\alpha, \beta)$, are equal to each other. In this case, using the resolution of the identity we show that
\begin{eqnarray}
&&|\langle C;\gamma, \delta\ket{C;\alpha, \beta}|^2=\frac{1}{d+1};\;\;\;(\gamma, \delta)\ne(\alpha, \beta).
\end{eqnarray}
This is called symmetric informationally complete positive operator-valued measures or SIC-POVM \cite{SIC1,SIC2,SIC3}). This problem is outside the scope of this paper.

Below we introduce $2d^2$ coherent states, so we have a more dense distribution of the states in Hilbert space.
From a practical point of view the use of  $2d^2$ coherent states (expansion in Eq.(\ref{A13}) below) is more immune in the presence of noise, than Eq.(\ref{A12}) which uses $d^2$ coherent states.
We discuss this in examples in section \ref{sec379}.

\section{The parity operator and Wigner and Weyl functions}

\subsection{The parity operator}

The parity operator around the origin in phase space\cite{11A,11,12,13}, is the unitary and Hermitian operator
\begin{eqnarray}
&&{\mathfrak P}=F^2;\;\;\;{\mathfrak P}^2={\bf 1};\;\;\;{\mathfrak P}={\mathfrak P}^\dagger;\;\;\;[{\mathfrak P},F]=0\\
&&{\mathfrak P}\ket{X;j}=\ket{X;-j};\;\;\;{\mathfrak P}\ket{P;j}=\ket{P;-j}.
\end{eqnarray}
The parity operator transforms the position and momentum operators ${\widehat x}$, ${\widehat p}$, as
\begin{eqnarray}
{\mathfrak P}{\widehat x}{\mathfrak P}=-{\widehat x};\;\;\;{\mathfrak P}{\widehat p}{\mathfrak P}=-{\widehat p}.
\end{eqnarray}
As we explained earlier we work with exponentials of these quantities which are single-valued and we get
\begin{eqnarray}
{\mathfrak P}X{\mathfrak P}=X^{-1};\;\;\;{\mathfrak P}Z{\mathfrak P}=Z^{-1}.
\end{eqnarray}
More generally
\begin{eqnarray}
{\mathfrak P}^\nu D(\alpha, \beta, \gamma){\mathfrak P}^\nu =D[(-1)^\nu\alpha, (-1)^\nu \beta,\gamma];\;\;\;\nu\in {\mathbb Z}(2).
\end{eqnarray}

We note that even in classical physics, displacements and the parity transformation do not commute. When displacement is applied first and parity transformation afterwards, we get
\begin{eqnarray}\label{C1}
x\xrightarrow {\rm displacement} x+a\xrightarrow {\rm parity}-(x+a)
\end{eqnarray}
When parity transformation is applied first and displacement afterwards, we get
\begin{eqnarray}\label{C2}
x\xrightarrow {\rm parity}-x\xrightarrow {\rm displacement}-x+a.
\end{eqnarray}

In the case of odd $d$, the eigenvalues of ${\mathfrak P}$ are $1$ (with multiplicity $\frac{d+1}{2}$) and $-1$ (with multiplicity $\frac{d-1}{2}$). The eigenvectors corresponding to the eigenvalue $1$ span a $\left (\frac{d+1}{2}\right )$-dimensional subspace of $H(d)$ which we denote as $H_0\left (\frac{d+1}{2}\right )$, and the 
eigenvectors corresponding to the eigenvalue $-1$ span a $\left (\frac{d-1}{2}\right )$-dimensional subspace of $H(d)$ which we denote as $H_1\left (\frac{d-1}{2}\right )$.
Then
\begin{eqnarray}
H(d)=H_0\left (\frac{d+1}{2}\right )\oplus H_1\left (\frac{d-1}{2}\right ).
\end{eqnarray}
Orthonormal bases in $H_0\left (\frac{d+1}{2}\right )$ and $H_1\left (\frac{d-1}{2}\right )$ are
\begin{eqnarray}
\left\{\ket {X;0}, \frac{1}{\sqrt{2}}(\ket{X;j}+\ket{X;-j})\;|\;j=1,...,\frac{d-1}{2}\right \};\;\;\;\left \{\frac{1}{\sqrt{2}}(\ket{X;j}-\ket{X;-j})\;|\;j=1,...,\frac{d-1}{2}\right \}.
\end{eqnarray}
We call $\varpi_0$, $\varpi_1$ the projectors to these subspaces.Then
\begin{eqnarray}\label{158}
\varpi_0-\varpi_1={\mathfrak P};\;\;\;\varpi_0+\varpi_1={\bf 1}.
\end{eqnarray}
If $\rho_0$ is a density matrix, the expectation value of the observable (Hermitian matrix) ${\mathfrak P}$ is
\begin{eqnarray}\label{158A}
<{\mathfrak P}>={\rm Tr}({\mathfrak P}\rho_0)\in [-1,1].
\end{eqnarray}

Parity is a symmetry of a  Hamiltonian $h$, when $h$ commutes with ${\mathfrak P}$:
\begin{eqnarray}
[{\mathfrak P},h]=0.
\end{eqnarray}
If the density matrix $\rho_0$ evolves in time with a Hamiltonian $h$ as $\rho(t)=\exp(iht)\rho_0\exp(-iht)$, then
\begin{eqnarray}
\frac{d<{\mathfrak P}>}{dt}=i{\rm Tr}([{\mathfrak P},h]\rho_0).
\end{eqnarray}
If parity is a symmetry of $h$, then $\frac{d<{\mathfrak P}>}{dt}=0$.

The displaced parity operator (i.e., parity operator around a point $(\alpha, \beta) \in PS(X,Z)$ ), is the unitary operator
\begin{eqnarray}\label{201}
&&{\mathfrak P}(\alpha, \beta)=D(\alpha, \beta,\gamma){\mathfrak P}[D(\alpha, \beta,\gamma)]^\dagger=D(2\alpha, 2\beta,0){\mathfrak P}={\mathfrak P}D(-2\alpha, -2\beta,0)\nonumber\\
&&[{\mathfrak P}(\alpha, \beta)]^2={\bf 1}.
\end{eqnarray}
Clearly ${\mathfrak P}(0,0)={\mathfrak P}$. 
Also ${\mathfrak P}(\alpha, 0)=Z^{2\alpha}{\mathfrak P}$ are displaced parity operators in the momentum direction, and
${\mathfrak P}(0,\beta)=X^{2\beta}{\mathfrak P}$ are displaced parity operators in the position direction.

\begin{example}
For $d=3$ 
\begin{eqnarray}\label{30AB}
\ket{X;-1}=\begin{pmatrix}
1\\0\\0
\end{pmatrix};\;\;\;
\ket{X;0}=\begin{pmatrix}
0\\1\\0
\end{pmatrix};\;\;\;
\ket{X;1}=\begin{pmatrix}
0\\0\\1
\end{pmatrix}.
\end{eqnarray}
Then the matrices $Z,X,F, {\mathfrak P}$ are
\begin{eqnarray}\label{700}
Z={\rm diag}\begin{pmatrix}
\omega(-1)&1&\omega(1)
\end{pmatrix};\;\;\;
X=\begin{pmatrix}
0&0&1\\
1&0&0\\
0&1&0
\end{pmatrix};\;\;\;
F=\frac{1}{\sqrt{3}}\begin{pmatrix}
\omega(1)&1&\omega(-1)\\
1&1&1\\
\omega(-1)&1&\omega(1)
\end{pmatrix};\;\;\;{\mathfrak P}=\begin{pmatrix}
0&0&1\\
0&1&0\\
1&0&0
\end{pmatrix}.
\end{eqnarray}
We used Eq.(\ref{HAM}) to calculate the Hamiltonian that will give stroboscopically the time evolution operator $D(1,1,0)=ZX\omega(-2)$:
\begin{eqnarray}
h=\frac{3}{2\pi i}\log D(1,1,0)=\begin{pmatrix}
0& 0.50 - 0.28i&   -0.57i\\
  0.50 + 0.28i& 0& -0.50-0.28i\\
   0.57i & - 0.50+0.28i& 0
\end{pmatrix}.
\end{eqnarray}
In this case $[{\mathfrak P},h]\ne 0$ and parity symmetry is violated.

\end{example}

\subsection{Properties of the displacement and parity operators}

Here we summarise briefly known (e.g.\cite{1}) properties of the displacement and parity operators, because they are generalised in the present paper later. 
Their proof is based on getting the matrix elements of both sides with the $\bra{X;k}$ and $\ket{X;j}$, using Eq.(\ref{122}):
\begin{itemize}\item
The $D(\alpha, \beta, 0)$ and ${\mathfrak P}(\gamma, \delta)$ are related through a Fourier transform:
\begin{eqnarray}\label{KKLL}
\frac {1}{d}\sum _{\alpha ,\beta } D(\alpha, \beta,0) \omega(\beta \gamma -\alpha \delta)&=&{\mathfrak P}(\gamma,\delta).
\end{eqnarray}
The inverse Fourier transform gives
\begin{eqnarray}\label{KKLL1}
\frac {1}{d}\sum _{\gamma,\delta } {\mathfrak P}(\gamma,\delta) \omega(-\beta \gamma +\alpha \delta)&=&D(\alpha, \beta,0).
\end{eqnarray}
\item
If $D_{ij}(\alpha, \beta,0)=\bra{X;i}D(\alpha, \beta,0) \ket{X;j}$ and  ${\mathfrak P}_{ij}(\alpha, \beta)=\bra{X;i}{\mathfrak P}(\alpha, \beta) \ket{X;j}$, then
\begin{eqnarray}\label{DFF}
\frac {1}{d}\sum _{\alpha,\beta}D_{ij}(\alpha, \beta,0)D_{k\ell}(-\alpha, -\beta,0)=\frac {1}{d}\sum _{\alpha,\beta}{\mathfrak P}_{ij}(\alpha, \beta){\mathfrak P}_{k\ell}(\alpha, \beta)=\delta_{i,\ell}\delta_{j,k}.
\end{eqnarray}
\item
The displacement operators obey the marginal relations
\begin{eqnarray}\label{AX1}
&&\frac {1}{d}\sum _{\beta } D(\alpha, \beta,0) =\ket{P;2^{-1}\alpha}\bra {P;2^{-1}\alpha}{\mathfrak P}\nonumber\\
&&\frac {1}{d}\sum _{\alpha} D(\alpha, \beta,0) =\ket{X;2^{-1}\beta}\bra {X;2^{-1}\beta}{\mathfrak P}\nonumber\\
&&\frac {1}{d}\sum _{\alpha,\beta} D(\alpha, \beta,0) ={\mathfrak P};\;\;\;(\alpha, \beta) \in PS(X,Z),
\end{eqnarray}
and the displaced parity operators obey the marginal relations
\begin{eqnarray}\label{AX2}
&&\frac {1}{d}\sum _{\beta } {\mathfrak P}(\alpha, \beta) =\ket{P;\alpha}\bra {P;\alpha}\nonumber\\
&&\frac {1}{d}\sum _{\alpha} {\mathfrak P}(\alpha, \beta) =\ket{X;\beta}\bra {X;\beta}\nonumber\\
&&\frac {1}{d}\sum _{\alpha, \beta} {\mathfrak P}(\alpha, \beta) ={\bf 1};\;\;\;(\alpha, \beta) \in PS(X,Z).
\end{eqnarray}
\end{itemize}
The Fourier transform between displacement operators and displaced parity operators in Eqs.(\ref{KKLL}),(\ref{KKLL1}), and also the analogy (`duality') between Eqs.(\ref{AX1}),(\ref{AX2})
suggest that the displacement operators and the displaced parity operators might be  `two sides of the same coin'.
The unification of these quantities and also of the Weyl and Wigner functions which are intimately related to them, is one of the objectives of this paper.

\subsection{Wigner and Weyl functions}
There is a lot of work on Wigner functions in a discrete phase space(\cite{5,6,6AA,6A,6B,6C,6D,6E,6F,6G}) and there are different definitions (especially for systems with even dimension $d$). 
We consider the case of odd $d$, and define Wigner and Weyl functions in terms of displaced parity operators and the displacement operators correspondingly. This is analogous to the  Wigner and Weyl functions
in the case of continuous variables. 

Let $\Theta$ be an operator.
Its Wigner function $W(\Theta; \alpha ,\beta )$ and its Weyl (or ambiguity) function $\widetilde W(\Theta; \alpha ,\beta )$ are defined as
\begin{eqnarray}\label{WIG}
W(\Theta; \alpha ,\beta )={\rm Tr}[\Theta {\mathfrak P}(\alpha, \beta)];\;\;\;\widetilde W(\Theta; \alpha ,\beta )\equiv {\rm Tr}[\Theta D(\alpha, \beta,0)];\;\;\;(\alpha, \beta) \in PS(X,Z).
\end{eqnarray}
If $\Theta$ is a Hermitian operator, the Wigner function is real.

We summarise briefly some properties which follow from the properties in the previous  section (e.g.\cite{1}). They are generalised later in section \ref{sec18}.
\begin{itemize}
\item
The Wigner and Weyl functions are related through Fourier transforms (this follows immediately from Eqs(\ref{KKLL}), (\ref{KKLL1})):
\begin{eqnarray}\label{QW10}
\frac {1}{d}\sum _{\alpha ,\beta }  {\widetilde W}(\Theta; \alpha, \beta)\omega(\beta \gamma -\alpha \delta)&=&W(\Theta; \gamma,\delta),
\end{eqnarray}
and
\begin{eqnarray}\label{QW11}
\frac {1}{d}\sum _{\gamma,\delta } W(\Theta; \gamma,\delta) \omega(-\beta \gamma +\alpha \delta)&=&{\widetilde W}(\Theta;\alpha, \beta).
\end{eqnarray}
\item
From Eq.(\ref{DFF}) follows that an operator $\Theta$ can be expanded in terms of the $d^2$ displacement operators $D(\alpha, \beta, 0)$ (elements of $HW(d)/{\mathbb Z}(d)$ group) with the Weyl functions as coefficients.  It can also be expanded in terms of the $d^2$ displaced parity operators ${\mathfrak  P}(\alpha, \beta)$, with the Wigner functions as coefficients:
\begin{eqnarray}\label{exp}
&&\Theta=\frac{1}{d}\sum _{\alpha ,\beta} \widetilde W(\Theta; -\alpha ,-\beta ) D(\alpha, \beta, 0)\nonumber\\
&&\Theta=\frac{1}{d} \sum _{\alpha ,\beta} W(\Theta; \alpha ,\beta ){\mathfrak  P}(\alpha, \beta).
\end{eqnarray}
We note that there is no redundancy in these expansions. The operator $\Theta$  is a $d\times d$ matrix with $d^2$ complex elements, and is represented with the $d^2$ complex numbers
$\widetilde W(\Theta; -\alpha ,-\beta )$ (or ${W}(\Theta; \alpha ,\beta ) $). In contrast, there is redundancy in the expansion in Eq.(\ref{300}) below.

\item
In the first expansion in Eq.(\ref{exp}), the $D(\alpha, \beta, 0)$ are elements of a group (the $HW(d)/{\cal G}_d({\bf 1})$ group), and therefore if we multiply two operators $\Theta_1\Theta_2$ we get an expansion of the same type:
\begin{eqnarray}
\Theta_1\Theta_2&=&\frac{1}{d^2}\sum _{\alpha _1,\beta _1, \alpha _2,\beta_2} \widetilde W(\Theta_1; -\alpha _1,-\beta _1)
\widetilde W(\Theta_2; -\alpha _2,-\beta _2)\nonumber \\
&\times &\omega[2^{-1}(\alpha_1\beta_2-\alpha_2\beta_1)]D(\alpha_1+\alpha_2, \beta_1+\beta_2, 0).
\end{eqnarray}
We change the variables $\alpha_1, \alpha_2$ into $\alpha=\alpha_1+\alpha_2, \alpha_2$ (and similarly for $\beta$), we get 
\begin{eqnarray}
\widetilde W(\Theta_1\Theta_2; -\alpha ,-\beta )&=&\frac{1}{d}\sum _{\alpha_2, \beta_2}
\omega ( 2^{-1}\alpha \beta_2-2^{-1}\alpha_2 \beta)
\widetilde W(\Theta_1; \alpha_2 -\alpha,\beta_2-\beta) 
\widetilde W(\Theta _2; -\alpha_2, -\beta_2).
\end{eqnarray}
In the second expansion in Eq.(\ref{exp}) the ${\mathfrak  P}(\alpha, \beta)$ do not form a group, and if we multiply two operators $\Theta_1\Theta_2$ we do not get an expansion of the same type.
The Wigner function of the product $\Theta_1\Theta_2$ is given by the Moyal star product\cite{M1,M2,M3}, 
which in the context of quantum systems with finite-dimensional Hilbert space
is given by
\begin{eqnarray}\label{Moyal}
[W(\Theta _1)\star W(\Theta _2)] (\alpha ,\beta )\equiv
W(\Theta _1\Theta _2; \alpha ,\beta )&=&\frac{1}{d^2}\sum _{\alpha _1,\beta _1,\alpha _2,\beta _2}
\omega (2\alpha _2 \beta _1-2\alpha _1\beta _2)\nonumber\\ &\times&
 W(\Theta _1; \alpha +\alpha _1,\beta +\beta _1)W(\Theta _2; \alpha +\alpha _2,\beta +\beta _2) 
\end{eqnarray}
The Moyal star product is associative, and (in general) non-commutative.

\item
From Eqs(\ref{AX2}) follow the following marginal properties of the Wigner function:
\begin{eqnarray}\label{AX22}
&&\frac {1}{d}\sum _{\beta } W(\Theta; \alpha, \beta) =\bra {P;\alpha}\Theta \ket{P;\alpha}\nonumber\\
&&\frac {1}{d}\sum _{\alpha} W(\Theta; \alpha, \beta) =\bra {X;\beta}\Theta \ket{X;\beta}\nonumber\\
&&\frac {1}{d}\sum _{\alpha, \beta} W(\Theta; \alpha, \beta) ={\rm Tr}(\Theta);\;\;\;(\alpha, \beta) \in PS(X,Z).
\end{eqnarray}
If $\Theta$ is a density matrix, the right hand sides of the first two equations, are probabilities.
Also from Eqs(\ref{AX1}) follow the following marginal properties of the Weyl function:
\begin{eqnarray}\label{AX11}
&&\frac {1}{d}\sum _{\beta } {\widetilde W}(\Theta; \alpha, \beta) =\bra {P;2^{-1}\alpha}{\mathfrak P}\Theta \ket{P;2^{-1}\alpha}\nonumber\\
&&\frac {1}{d}\sum _{\alpha} {\widetilde W}(\Theta; \alpha, \beta) =\bra {X;2^{-1}\beta}{\mathfrak P}\Theta \ket{X;2^{-1}\beta}\nonumber\\
&&\frac {1}{d}\sum _{\alpha,\beta} {\widetilde W}(\Theta; \alpha, \beta) ={\rm Tr}({\mathfrak P}\Theta);\;\;\;(\alpha, \beta) \in PS(X,Z).
\end{eqnarray}

\end{itemize}

\section{Unification of displacement and parity transformations with the dihedral group}

\subsection{The dihedral group $\Delta_d$}
The dihedral group $\Delta_d$ is the group of symmetries of a regular $d$-sided polygon (some authors use the notation $\Delta_{2d}$).
It has $2d$ elements $x_a$ (rotations), and $p_a$  (reflections), where the index $a$ takes values in ${\mathbb Z}(d)$.
The unit element is $x_0={\bf 1}$. Multiplication is defined as
\begin{eqnarray}
x_a x_b=x_{a+b};\;\;\;x_a p_b=p_{a+b};\;\;\;p_ax_b=p_{a-b};\;\;\;p_ap_b=x_{a-b};\;\;\;a,b\in {\mathbb Z}(d).
\end{eqnarray}
It is easily seen that 
\begin{eqnarray}
x_a=x_1^a;\;\;\;x_1^d={\bf 1};\;\;\;p_a^2={\bf 1};\;\;\;x_a^{-1}=x_{-a};\;\;\;p_a^{-1}=p_a.
\end{eqnarray}
In this paper we use odd $d$ with $d\ge 3$ and then the group $\Delta_d$ is non-Abelian.
Rotations commute with each other ($x_ax_b=x_bx_a$), but in general reflections do not commute with each other ($p_ap_b\ne p_bp_a$), and also in general rotations do not commute with reflections($x_ap_b\ne p_bx_a$).

All elements can be generated from $x_1, p_1$ as follows:
\begin{eqnarray}
x_a=x_1^a;\;\;\;p_a=x_1^{a-1}p_1.
\end{eqnarray}
The general element of  $\Delta_d$ can be written as $R(a,\nu)=x_1^a p_1^\nu=p_1^\nu x_1^{(-1)^\nu a}$, where $a\in {\mathbb Z}(d)$ and $\nu\in{\mathbb Z}(2)$.
$\nu=0$ corresponds to rotations, and $\nu=1$ to reflections.
The multiplication rule is
\begin{eqnarray}\label{32}
R(a_1,\nu_1)R(a_2,\nu_2)=R[a_1+(-1)^{\nu_1}a_2,\nu_1+\nu_2].
\end{eqnarray}
The unit element is $R(0,0)={\bf 1}$. The inverse of $R(a,\nu)$ is
\begin{eqnarray}\label{33}
[R(a,\nu)]^{-1}=R[(-1)^{\nu+1}a, -\nu]=R[(-1)^{\nu+1}a, \nu].
\end{eqnarray}
It is easily seen that
\begin{eqnarray}
[R(a,0)]^{d}=[R(a,1)]^2={\bf 1},
\end{eqnarray}
and therefore in general
\begin{eqnarray}
[R(a,\nu)]^{2d}={\bf 1}.
\end{eqnarray}

\begin{proposition}
The dihedral group is the semidirect product of ${\mathbb Z}(d)$ by ${\mathbb Z}(2)$:
\begin{eqnarray}\label{18}
\Delta_d={\mathbb Z}(d) \rtimes {\mathbb Z}(2).
\end{eqnarray}
\end{proposition}
\begin{proof}
The proof is well known. But  we briefly sketch the proof because it is the first step towards a similar proof for $HWP(d)$ later (in proposition \ref{pro123}).
It involves the following steps:
\begin{itemize}
\item

The $\{R(a,0)=x_1^a|a\in {\mathbb Z}(d)\}$ form a normal subgroup of $\Delta_d$ isomorphic to ${\mathbb Z}(d)$.
Indeed
\begin{eqnarray}
R(a,\nu)R(b,0)[R(a,\nu)]^{-1}=R[(-1)^\nu b, 0].
\end{eqnarray}
\item
The $\{{\bf 1}, R(0,1)=p_1\}$ form a subgroup of $\Delta_d$ isomorphic to ${\mathbb Z}(2)$.
\item
Every element in $\Delta_d$ is a product of an element of ${\mathbb Z}(d)$ times an element of ${\mathbb Z}(2)$ and ${\mathbb Z}(d)\cap{\mathbb Z}(2)=\{{\bf 1}\}$.
\end{itemize}
This completes the proof.
\end{proof}

Physically, the semidirect product enlarges the group ${\mathbb Z}(d)$ of rotations, with reflections. 
A natural symmetry of a regular $d$-sided polygon is the rotations, associated with ${\mathbb Z}(d)$.
The semidirect product brings some extra symmetry related to reflections.

\subsection{A representation of the dihedral group with $Z$ and ${\mathfrak P}$ operators}

We introduce a representation of the dihedral group, with the following correspondence: 
\begin{eqnarray}
x_1\rightarrow Z;\;\;\;p_1\rightarrow {\mathfrak P}.
\end{eqnarray}
Then $x_a$ and $p_b$ correspond to the displacement operators and displaced parity operators, as follows:
\begin{eqnarray}
x_a\rightarrow Z^a=D(a,0,0);\;\;\;p_b\rightarrow Z^{b-1}{\mathfrak P}={\mathfrak P}[2^{-1}(b-1),0].
\end{eqnarray}
All elements can be generated from $Z$ and ${\mathfrak P}$. The general element is
\begin{eqnarray}
{\mathfrak Z}(a,\nu)=Z^{a}{\mathfrak P}^\nu={\mathfrak P}^\nu Z^{(-1)^\nu a};\;\;\;a\in {\mathbb Z}(d);\;\;\;\nu \in{\mathbb Z}(2).
\end{eqnarray}
$\nu=0$ corresponds to displacements with $Z$ in the momentum direction, and $\nu=1$ to displaced parity transformations with $Z^a{\mathfrak P}={\mathfrak P}[2^{-1}a,0]$. 

Relations analogous to the ones in the previous section hold here also:
\begin{eqnarray}\label{abc}
&&[{\mathfrak Z}(a,\nu)]^{-1}={\mathfrak Z}[(-1)^{\nu+1}a, \nu]={\mathfrak Z}[(-1)^{\nu+1}a, -\nu]\nonumber\\
&&{\mathfrak Z}(a_1,\nu_1){\mathfrak Z}(a_2,\nu_2)={\mathfrak Z}[a_1+(-1)^{\nu_1}a_2,\nu_1+\nu_2]\nonumber\\
&&[{\mathfrak Z}(a,0)]^{d}={\bf 1};\;\;\;[{\mathfrak Z}(a,1)]^{2}={\bf 1};\;\;\;[{\mathfrak Z}(a,\nu)]^{2d}={\bf 1}.
\end{eqnarray}

There is an analogy between the symmetries of a regular $d$-sided polygon described with the dihedral group $\Delta _d$ and the symmetries of the momentum states in $H(d)$:
\begin{itemize}
\item
Vertices of a regular $d$-sided polygon correspond to the momentum states $\ket{P;j}$.
\item
Rotations of the vertices with $x_a$, correspond to acting with the displacement operators $Z^a$ on the momentum states $\ket{P;j}$.
\item
Reflections with $p_b$ on the vertices, correspond to acting with the displaced parity operators ${\mathfrak P}[2^{-1}a,0]$ (with $a=b-1$) on the momentum states $\ket{P;j}$:
\begin{eqnarray}
{\mathfrak P}[2^{-1}a,0]\ket{P;j}=Z^a{\mathfrak P}\ket{P;j}=\ket {P;-j+a}
\end{eqnarray}
\end{itemize}
We denote this representation as $\Delta_d(Z)$.
The group of displacements with $Z$ in the momentum direction (which is ${\mathbb Z}(d)$),
is enlarged with parity transformations. 
Displacements in the momentum direction do not commute with the parity transformations and $\Delta_d(Z)$ is non-Abelian.

Physically we can implement stroboscopically  ${\mathfrak Z}(a,\nu)$ transformations, with a Hamiltonian that is the principal logarithm of ${\mathfrak Z}(a,\nu)$:
\begin{eqnarray}
h=\frac{d}{2\pi i}\log {\mathfrak Z}(a,\nu).
\end{eqnarray}
Then stroboscopically we get
\begin{eqnarray}\label{59}
\exp(iht)={\mathfrak Z}(a,\nu);\;\;\;t=\frac{2\pi}{d}+4\pi N;\;\;\;N=0,1,...
\end{eqnarray}

From Eqs(\ref{AX1}), (\ref{AX2}) we get the following marginal property of ${\mathfrak Z}(a,\nu)$:
\begin{eqnarray}\label{60}
\frac{1}{d}\sum_a {\mathfrak Z}(a,\nu)=\ket{X;0}\bra{X;0}.
\end{eqnarray}
This can also be proved directly by taking the matrix elements with $\bra{X;j}$ and $\ket{X;k}$.

If $\Theta$ is an arbitrary operator we define the 
\begin{eqnarray}
{\mathfrak W}_Z(\Theta; a,\nu)={\rm Tr}[\Theta {\mathfrak Z}(a,\nu)].
\end{eqnarray}
Then ${\mathfrak W}_Z(\Theta; a,0)$ is equal to the Weyl function (with the second variable equal to zero),  and ${\mathfrak W}_Z(\Theta; a,1)$ is equal to the Wigner function (with the second  variable equal to zero):
\begin{eqnarray}
{\mathfrak W}_Z(\Theta; a,0)={\widetilde W}(\Theta, a,0);\;\;\;{\mathfrak W}_Z(\Theta; a,1)={W}(\Theta; a,0).
\end{eqnarray}
The ${\mathfrak W}_Z(\Theta; a,\nu)$ unifies the Weyl and Wigner functions (with the second variable equal to zero), and we call it `Wigner-Weyl function'.
The full unification is discussed in section \ref{sec18}.

From Eq.(\ref{60}) follows the following marginal property of ${\mathfrak W}_Z(\Theta; a,\nu)$:
\begin{eqnarray}\label{139}
\frac{1}{d}\sum _a{\mathfrak W}_Z(\Theta; a,\nu )=\bra{X;0}\Theta \ket{X;0}.
\end{eqnarray}

\begin{example}\label{ex1}
In $H(d)$ let  $\rho_1=\ket{X;j}\bra{X;j}$. In this case we get
\begin{eqnarray}
{\mathfrak W}_Z(\rho_1; a, \nu)=\omega [a j(-1)^\nu]\delta(2\nu j,0).
\end{eqnarray}
We also consider the density matrix $\rho_2=\ket{P;j}\bra{P;j}$. In this case we get
\begin{eqnarray}
{\mathfrak W}_Z(\rho_2; a, \nu)=\delta(2\nu j,a).
\end{eqnarray}
In both cases the marginal relation in Eq.(\ref{139}) holds.

\end{example}

\subsection{$\Delta_d(Z)$ as solvable group and commutators that perform displacements and parity transformations along loops in phase space}\label{sec34}

\begin{lemma}
The commutator of elements of the $\Delta_d(Z)$ group is 
\begin{eqnarray}\label{34}
&&{\cal L}_2(a_1,\nu_1|a_2,\nu_2)=[{\mathfrak Z}(a_1,\nu_1),{\mathfrak Z}(a_2,\nu_2)]=Z^{2{\cal A}(a_1,\nu_1|a_2,\nu_2)}\nonumber\\
&&{\cal A}(a_1,\nu_1|a_2,\nu_2)=a_1\nu_2-a_2\nu_1\in {\mathbb Z}(d);\;\;\;\nu_1,\nu_2=0,1.
\end{eqnarray}

\end{lemma}
\begin{proof}

The commutator of elements of $\Delta_d(Z)$ is 
\begin{eqnarray}
&&{\cal L}_2(a_1,\nu_1|a_2,\nu_2)={\mathfrak Z}(a_1,\nu_1){\mathfrak Z}(a_2,\nu_2)[{\mathfrak Z}(a_1,\nu_1)]^{-1}[{\mathfrak Z}(a_2,\nu_2)]^{-1}={\mathfrak Z}(A,0)=Z^A\nonumber\\
&&A=[1-(-1)^{\nu_2}]a_1-[1-(-1)^{\nu_1}]a_2=2(\nu_2a_1-\nu_1a_2);\;\;\;\nu_1,\nu_2=0,1.
\end{eqnarray}
In order to get a simple formula we used the following relation which holds only for $\nu=0,1$: 
\begin{eqnarray}\label{40}
(-1)^\nu=1-2\nu;\;\;\;\nu=0,1.
\end{eqnarray}
The $\nu_1,\nu_2$ take the values $0,1$ and $a_1,a_2$ take values in ${\mathbb Z}(d)$, so that ${\cal A}(a_1,\nu_1|a_2,\nu_2)\in {\mathbb Z}(d)$.
We note  that for odd $d$, the $2$ is invertible and therefore the $2(\nu_2a_1-\nu_1a_2)$ can take all values in ${\mathbb Z}(d)$. 

\end{proof}

\begin{proposition}
\mbox{}
\begin{itemize}
\item[(1)]
\begin{eqnarray}\label{77}
[\Delta_d(Z),\Delta_d(Z)]={\cal G}_d(Z);\;\;\;[{\cal G}_d(Z), {\cal G}_d(Z)]\cong \{{\bf 1}\}.
\end{eqnarray}
${\cal G}_d(Z)$ has been defined in Eq.(\ref{sub}).
\item[(2)]
$\Delta_d(Z)$ is a solvable group (of solvability class $2$) with derived series
\begin{eqnarray}\label{78}
\Delta_d(Z)\triangleright {\cal G}_d(Z)\triangleright \{{\bf 1}\}.
\end{eqnarray}
Here  the ${\cal G}_d(Z)$ is related to the non-commutativity between $Z$ and ${\mathfrak P}$ 
(as we explained in Eqs(\ref{C1}),(\ref{C2}) displacements do not commute with parity even in classical physics).
\item[(3)]
Here the Abelian groups in Eq.(\ref{100}) are
\begin{eqnarray}\label{151}
\Delta_d(Z)/{\cal G}_d(Z)\cong {\mathbb Z}(2);\;\;\;{\cal G}_d(Z)/\{{\bf 1}\}={\cal G}_d(Z)\cong {\mathbb Z}(d).
\end{eqnarray}
\end{itemize}
\end{proposition}
\begin{proof}
\mbox{}
\begin{itemize}
\item[(1)]

The commutators ${\cal L}_2(a_1,\nu_1|a_2,\nu_2)$ in Eq.(\ref{34}), generate the group ${\cal G}_d(Z)$.
This proves the first of Eqs(\ref{77}).
Also ${\cal G}_d(Z)\cong {\mathbb Z}(d)$ is an Abelian group, and therefore its commutator is $\{{\bf 1}\}$.
\item[(2)]
From Eq.(\ref{77}) follows immediately that $\Delta_d(Z)$ is a solvable group with derived series given in Eq.(\ref{78}).
\item[(3)]
The fact that $\Delta_d(Z)$ is the semidirect product ${\mathbb Z}(d)\rtimes {\mathbb Z}(2)$ implies that $\Delta_d(Z)/{\mathbb Z}(d)\cong {\mathbb Z}(2)$.
The relation ${\cal G}_d(Z)/\{{\bf 1}\}={\cal G}_d(Z)\cong {\mathbb Z}(d)$ is straightforward.
\end{itemize}

\end{proof}

Acting with the commutator ${\cal L}_2(a_1,\nu_1|a_2,\nu_2)$ on a state $\ket{f}$, displaces it along a loop 
in ${\mathbb Z}(d)\times {\mathbb Z}(2)$, which is the momentum-parity phase space and is denoted as $PS(Z,{\mathfrak P})$ (phase space for $Z,{\mathfrak P}$).
The ${\cal A}(a_1,\nu_1|a_2,\nu_2)$ is `area' related to the exterior product of the vectors $(a_1,\nu_1)$ and $(a_2,\nu_2)$ defining this loop.  We get
\begin{eqnarray}\label{340}
{\cal L}_2(a_1,\nu_1|a_2,\nu_2)\ket{f}=Z^{2{\cal A}(a_1,\nu_1|a_2,\nu_2)}\ket{f}.
\end{eqnarray}
Then in general
\begin{eqnarray}\label{341}
|\bra{f}{\cal L}_2(a_1,\nu_1|a_2,\nu_2)\ket{ f}|=|\bra{f}Z^{2{\cal A}(a_1,\nu_1|a_2,\nu_2)}\ket{f}|\le 1.
\end{eqnarray}
This should be compared with Eq.(\ref{loop}).

\subsection{A representation of the dihedral group with $X$ and ${\mathfrak P}$ operators}

If $U$ is a unitary operator, we get another representation of $\Delta_d$ as
\begin{eqnarray}
x_a\rightarrow UZ^aU^\dagger;\;\;\;p_a\rightarrow UZ^{a-1}{\mathfrak P}U^\dagger,
\end{eqnarray}
In particular, we take $U$ to be the Fourier operator $F$. Then
\begin{eqnarray}
x_a\rightarrow F^\dagger Z^aF=X^a;\;\;\;p_a\rightarrow F^\dagger Z^{a-1}{\mathfrak P}F=X^{a-1}{\mathfrak P},
\end{eqnarray}
We denote this as $\Delta_d(X)$.
It combines displacements in the position direction, with parity transformations.

The analogy between the quantities in the previous section and those in the present section is:
\begin{eqnarray}
&&{\mathfrak Z}(a, \nu)\rightarrow {\mathfrak X}(b, \nu)=X^b{\mathfrak P}^\nu,\nonumber\\
&&PS(Z,{\mathfrak P})\rightarrow PS(X,{\mathfrak P}),\nonumber\\
&&[\Delta_d(Z),\Delta_d(Z)]={\cal G}_d(Z)\rightarrow [\Delta_d(X),\Delta_d(X)]={\cal G}_d(X),\nonumber\\
&&{\mathfrak W}_Z(\Theta; a, \nu)={\rm Tr}[\Theta {\mathfrak Z}(a,\nu)]\rightarrow {\mathfrak W}_X(\Theta; b, \nu)={\rm Tr}[\Theta {\mathfrak X}(b,\nu)].
\end{eqnarray}
Then
\begin{eqnarray}
{\mathfrak Z}(a, \nu)=F{\mathfrak X}(a, \nu)F^\dagger.
\end{eqnarray}
Also
\begin{eqnarray}
{\mathfrak W}_X(\Theta; b,0)={\widetilde W}(\Theta; 0,b);\;\;\;{\mathfrak W}_X(\Theta; b,1)={W}(\Theta; 0,b) 
\end{eqnarray}

This leads naturally to the next section, where we unify the ${\mathfrak Z}(a, \nu)$ with ${\mathfrak X}(b, \nu)$, and to section \ref{sec18} where we unify ${\mathfrak W}_Z(\Theta; a, \nu)$ with 
${\mathfrak W}_X(\Theta; b, \nu)$.

\section{Heisenberg-Weyl-parity group: a generalised dihedral group}
\begin{definition}
The `displacement-parity operators'  ${\mathfrak D}(\alpha, \beta, \gamma, \nu)$ are the unitary operators
\begin{eqnarray}
{\mathfrak D}(\alpha, \beta, \gamma, \nu)=D(\alpha, \beta, \gamma){\mathfrak P}^\nu={\mathfrak P}^\nu D[(-1)^\nu\alpha, (-1)^\nu\beta, \gamma];\;\;\;\nu\in{\mathbb Z}(2);\;\;\;\alpha, \beta, \gamma\in {\mathbb Z}(d).
\end{eqnarray}
\end{definition}
Then:
\begin{itemize}
\item
${\mathfrak D}(\alpha, \beta, \gamma, \nu)$ unifies many of the operators introduced in previous sections:
\begin{eqnarray}
&&{\mathfrak D}(\alpha, \beta, \gamma, 0)={D}(\alpha, \beta, \gamma);\;\;\;
{\mathfrak D}(\alpha, \beta, 0,1)={\mathfrak P}(2^{-1}\alpha,2^{-1} \beta)\nonumber\\
&&{\mathfrak D}(\alpha,0,0, \nu)={\mathfrak Z}(\alpha, \nu);\;\;\;{\mathfrak D}(0,\beta,0, \nu)={\mathfrak X}(\beta, \nu)
\end{eqnarray}
So the ${\mathfrak D}(\alpha, \beta, \gamma, 0)$ are the same as the displacement operators ${D}(\alpha, \beta, \gamma)$ in section \ref{sec19}.
We will show (Eq.(\ref{WW2})) that the extra displacement operators ${\mathfrak D}(\alpha, \beta, \gamma, 1)$ are related to ${D}(\alpha, \beta, \gamma)$ through a Fourier transform.

\item
The inverse of ${\mathfrak D}(\alpha, \beta, \gamma, \nu)$ is
\begin{eqnarray}\label{inv}
[{\mathfrak D}(\alpha, \beta, \gamma, \nu)]^{-1}=[{\mathfrak D}(\alpha, \beta, \gamma, \nu)]^\dagger={\mathfrak P}^\nu D(-\alpha, -\beta, -\gamma)={\mathfrak D}[(-1)^{\nu+1}\alpha, (-1)^{\nu+1}\beta, -\gamma, \nu].
\end{eqnarray}
For $\nu=0,1$ we get
\begin{eqnarray}\label{inv0}
[{\mathfrak D}(\alpha, \beta, \gamma, 0)]^{-1}={\mathfrak D}(-\alpha, -\beta, -\gamma, 0);\;\;\;[{\mathfrak D}(\alpha, \beta, \gamma, 1)]^{-1}={\mathfrak D}(\alpha, \beta, -\gamma, 1).
\end{eqnarray}
\item
It is easily seen that
\begin{eqnarray}\label{inv1}
[{\mathfrak D}(\alpha, \beta, \gamma, 0)]^{d}=[{\mathfrak D}(\alpha, \beta, \gamma, 1)]^{2}={\bf 1},
\end{eqnarray}
and therefore in general
\begin{eqnarray}\label{inv10}
[{\mathfrak D}(\alpha, \beta, \gamma, \nu)]^{2d}={\bf 1}.
\end{eqnarray}
\item The trace of ${\mathfrak D}(\alpha, \beta, 0, \nu)$ is
\begin{eqnarray}
&&(\alpha, \beta)\ne (0,0)\;\Rightarrow\;{\rm Tr}[{\mathfrak D}(\alpha, \beta, 0, \nu)]=\delta (\nu,1)\nonumber\\
&&(\alpha, \beta)= (0,0)\;\Rightarrow\;{\rm Tr}[{\mathfrak D}(0,0, 0, \nu)]=d.
\end{eqnarray}
\item
The ${\mathfrak D}(\alpha, \beta, \gamma, \nu)$ act on position and momentum states as follows:
\begin{eqnarray}\label{190}
&&{\mathfrak D}(\alpha, \beta, \gamma,\nu)\ket{X;j}=\omega[2^{-1}\alpha\beta+\alpha j(-1)^\nu+\gamma]\ket{X;(-1)^\nu j+\beta}\nonumber\\
&&{\mathfrak D}(\alpha, \beta, \gamma,\nu)\ket{P;j}=\omega[-2^{-1}\alpha\beta-\beta j(-1)^\nu+\gamma]\ket{P;(-1)^\nu j+\alpha}.
\end{eqnarray}
\end{itemize}
\begin{proposition}\label{pro123}
\mbox{}
\begin{itemize}
\item[(1)]
The $2d^3$ displacement-parity operators with multiplication, form a non-Abelian group which we call Heisenberg-Weyl-parity group and denote as $HWP(d)$.
\item[(2)]
The $d^3$ operators ${\mathfrak D}(\alpha, \beta, \gamma, 0)$ form a normal subgroup of $HWP(d)$, isomorphic to the Heisenberg-Weyl group $HW(d)$.
\item[(3)]
$HWP(d)$ is the semidirect product
\begin{eqnarray}\label{semi}
HWP(d)=HW(d)\rtimes {\mathbb Z}(2).
\end{eqnarray}

\end{itemize}
\end{proposition}
\begin{proof}
\begin{itemize}
\item[(1)]
We first show closure under multiplication:
\begin{eqnarray}\label{46}
&&{\mathfrak D}(\alpha_1, \beta_1, \gamma_1, \nu_1){\mathfrak D}(\alpha_2, \beta _2, \gamma_2,\nu_2)=
{D}(\alpha_1, \beta_1, \gamma_1){\mathfrak P}^{\nu_1}{D}(\alpha_2, \beta _2, \gamma_2){\mathfrak P}^{\nu_2}\nonumber\\&&=
{D}(\alpha_1, \beta_1, \gamma_1){D}[(-1)^{\nu_1}\alpha_2, (-1)^{\nu_1}\beta _2, \gamma_2]{\mathfrak P}^{\nu_1+\nu_2}={\mathfrak D}(A,B,\Gamma, \nu_1+\nu_2),
\end{eqnarray}
where
\begin{eqnarray}
&&A=\alpha_1+(-1)^{\nu_1}\alpha_2;\;\;\;B=\beta_1+(-1)^{\nu_1}\beta_2;\;\;\;\nu_1,\nu_2=0,1\nonumber\\
&&\Gamma= \gamma_1+\gamma_2+2^{-1}(-1)^{\nu_1}(\alpha_1\beta _2-\alpha_2\beta_1).
\end{eqnarray}
When $\nu_1=\nu_2=0$ this is simply the multiplication rule in Eq(\ref{199F}).
Associativity holds. The unit element is ${\cal D}(0,0,0,0)={\bf 1}$. Inverses exist and are given in Eq.(\ref{inv}).
This proves that the $2d^3$ displacement-parity operators with multiplication, form a group.
\item[(2)]

The ${\mathfrak D}(\alpha, \beta, \gamma, 0)=D(\alpha, \beta, \gamma)$ are closed under multiplication and they are elements of the Heisenberg-Weyl group $HW(d)$, which
is a subgroup of $HWP(d)$. Also
\begin{eqnarray}
{\mathfrak D}(\alpha_1, \beta_1, \gamma_1, \nu_1){\mathfrak D}(\alpha, \beta , \gamma,0)[{\mathfrak D}(\alpha_1, \beta_1, \gamma_1, \nu_1)]^{-1}={\mathfrak D}(A,B,\Gamma,0),
\end{eqnarray}
where
\begin{eqnarray}
&&A=\alpha+2\nu_1(\alpha_1-\alpha);\;\;\;B=\beta+2\nu_1(\beta_1-\beta);\;\;\;\nu_1=0,1\nonumber\\
&&\Gamma=\gamma+(1-\nu_1)(\alpha_1\beta-\alpha\beta_1)
\end{eqnarray}
This proves that the Heisenberg-Weyl group $HW(d)$ is a normal subgroup of $HWP(d)$.
\item[(3)]
We have proved that $HW(d)$ is a normal subgroup of $HWP(d)$.
Also the 
\begin{eqnarray}
\{{\mathfrak D}(0,0,0,0)={\bf 1}, {\mathfrak D}(0,0,0,1)={\mathfrak P}\}\simeq{\mathbb Z}(2),
\end{eqnarray}
is a subgroup of $HWP(d)$. In addition to that every element of $HWP(d)$ can be written as a product of an element of $HW(d)$ and an element of ${\mathbb Z}(2)$:
\begin{eqnarray}
{\mathfrak D}(\alpha, \beta, \gamma, \nu)=D(\alpha, \beta, \gamma){\mathfrak P}^\nu.
\end{eqnarray}
Furthermore $HW(d)\cap {\mathbb Z}(2)=\{\bf 1\}$.
This proves that $HWP(d)$ is the semidirect product in Eq.(\ref{semi}).

\end{itemize}
\end{proof}
We note that $\Delta_d(X)$ and $\Delta_d(Z)$ are subgroups of $HWP(d)$.

We have enlarged the group $HW(d)$ of displacements in phase space, with parity transformations.
The semidirect product makes the parity operator an integral part of the Heisenberg-Weyl formalism.

Physically we can implement stroboscopically these transformations, with a Hamiltonian that is the principal logarithm of ${\mathfrak D}(\alpha, \beta, \gamma, \nu)$:
\begin{eqnarray}\label{HA}
h=\frac{d}{2\pi i}\log {\mathfrak D}(\alpha, \beta, \gamma, \nu).
\end{eqnarray}
The time evolution operator is  stroboscopically
\begin{eqnarray}\label{10B}
\exp(iht)={\mathfrak D}(\alpha, \beta, \gamma, \nu);\;\;\;t=\frac{2\pi}{d}+4\pi N;\;\;\;N=0,1,...
\end{eqnarray}

\begin{proposition}\label{pro1}
The following are properties of the $2d^2$ operators ${\mathfrak D}(\alpha, \beta,0,\nu) $ which are elements of the $HWP(d)/{\mathbb Z}(d)$ group 
(that contains the  ${\mathfrak D}(\alpha, \beta,\gamma,\nu) $ modulo a phase factor $\omega(\gamma)$):

\begin{itemize}
\item[(1)]
The ${\mathfrak D}(\alpha, \beta,0,\nu) $ is related to ${\mathfrak D}(\gamma, \delta, 0,\nu+1)$ through a Fourier transform:
\begin{eqnarray}\label{WW2}
\frac {1}{d}\sum _{\alpha ,\beta } {\mathfrak D}(\alpha, \beta,0,\nu) \omega[2^{-1}(-1)^\nu(\beta \gamma -\alpha \delta)]={\mathfrak D}(\gamma, \delta, 0,\nu+1);\;\;\;\nu\in{\mathbb Z}(2).
\end{eqnarray}
\item[(2)]
If ${\mathfrak D}_{ij}(\alpha, \beta,0,\nu)=\bra{X;i}{\mathfrak D}(\alpha, \beta,0,\nu) \ket{X;j}$, then
\begin{eqnarray}\label{350}
\frac {1}{2d}\sum_{\nu=0}^1\sum _{\alpha,\beta}{\mathfrak D}_{ij}[(-1)^{\nu+1}\alpha, (-1)^{\nu+1}\beta,0,\nu]{\mathfrak D}_{k\ell}(\alpha, \beta,0,\nu)=\delta_{i\ell}\delta_{jk}.
\end{eqnarray}

\item[(3)]
For any operator $\Theta$ with ${\rm Tr}(\Theta)\ne 0$, we get
\begin{eqnarray}\label{TR55}
\frac {1}{2d}\sum_{\nu=0}^1\sum _{\alpha, \beta } D(\alpha, \beta,0, \nu) \frac{\Theta}{{\rm Tr}(\Theta)} [D(\alpha, \beta,0, \nu)]^\dagger={\bf 1}.
\end{eqnarray}
\item[(4)]
The following  marginal relations hold for ${\mathfrak D}(\alpha, \beta,0,\nu)$:
\begin{eqnarray}\label{121}
&&\frac{1}{d}\sum_{\beta}{\mathfrak D}(\alpha, \beta,0,\nu)=\ket{P;2^{-1}\alpha}\bra{P;2^{-1}\alpha}{\mathfrak P}^{\nu+1};\;\;\;\nu\in{\mathbb Z}(2)\nonumber\\
&&\frac{1}{d}\sum_{\alpha}{\mathfrak D}(\alpha, \beta,0,\nu)=\ket{X;2^{-1}\beta}\bra{X;2^{-1}\beta}{\mathfrak P}^{\nu+1}\nonumber\\
&&\frac{1}{d}\sum_{\alpha,\beta}{\mathfrak D}(\alpha, \beta,0,\nu)={\mathfrak P}^{\nu+1}.
\end{eqnarray}
Additional marginal relations are
\begin{eqnarray}\label{70}
&&\frac{1}{2d}\sum_{\nu=0}^1\sum_{\alpha,\beta}{\mathfrak D}(\alpha, \beta,0,\nu)=\varpi_0\nonumber\\
&&\frac{1}{2d}\sum_{\nu=0}^1\sum_{\alpha,\beta}(-1)^{\nu+1} {\mathfrak D}(\alpha, \beta,0,\nu)=\varpi_1,
\end{eqnarray}
where $\varpi_0, \varpi_1$ are the projectors in Eq.(\ref{158}).

\end{itemize}
\end{proposition}
\begin{proof}
In the proofs below it is important that $2^{-1}$ exists (because $d$ is odd), and therefore when $\alpha$ takes all values in ${\mathbb Z}(d)$ then $2^{-1}\alpha$ (and also the $2\alpha$) take all values in ${\mathbb Z}(d)$.
\begin{itemize}
\item[(1)]
For $\nu=0,1$ Eq.(\ref{WW2}) reduces to Eqs (\ref{KKLL}),(\ref{KKLL1}) which have been given earlier.
We can also give a direct proof by calculating the matrix elements of both sides in the position basis, using Eq.(\ref{190}).
\item[(2)]
We rewrite the two expansions in Eq.(\ref{DFF}), in a unified way as
\begin{eqnarray}\label{DFF1}
\frac {1}{d}\sum _{\alpha,\beta}{\mathfrak D}_{ij}[(-1)^{\nu+1}\alpha, (-1)^{\nu+1}\beta,0,\nu]{\mathfrak D}_{k\ell}(\alpha, \beta,0,\nu)=\delta_{i\ell}\delta_{jk}.
\end{eqnarray}
This equation holds for both $\nu=0$ and $\nu=1$, and adding them with equal weights, we get Eq.(\ref{350}).
We get a more general relation if we use weights $\lambda$ and $1-\lambda$:
\begin{eqnarray}\label{DFF2}
&&\frac {1}{d}\sum_{\nu=0}^1\left [\lambda_\nu \sum _{\alpha,\beta}{\mathfrak D}_{ij}[(-1)^{\nu+1}\alpha, (-1)^{\nu+1}\beta,0,\nu]{\mathfrak D}_{k\ell}(\alpha, \beta,0,\nu)\right]=\delta_{i\ell}\delta_{jk}\nonumber\\
&&\lambda_0=\lambda;\;\;\;\lambda_1=1-\lambda.
\end{eqnarray}

\item[(3)]

We multiply both sides of Eq.(\ref{TR}) by the parity operator ${\mathfrak P}$, and we get
\begin{eqnarray}\label{TR14}
\frac {1}{d}\sum _{\alpha, \beta } {\mathfrak P}D(\alpha, \beta,0) \frac{\Theta}{{\rm Tr}(\Theta)} [D(\alpha, \beta,0)]^\dagger{\mathfrak P}={\bf 1}.
\end{eqnarray}
But ${\mathfrak P}D(\alpha, \beta,0)={\mathfrak D}(-\alpha, -\beta, 0, 1)$ and 
$[D(\alpha, \beta,0)]^\dagger{\mathfrak P}=[{\mathfrak D}(-\alpha, -\beta, 0, 1)]^\dagger$ (Eq.(\ref{inv0})) and therefore
\begin{eqnarray}\label{TR17}
\frac {1}{d}\sum _{\alpha, \beta }  {\mathfrak D}(\alpha, \beta, 0, 1)\frac{\Theta}{{\rm Tr}(\Theta)} [{\mathfrak D}(\alpha, \beta, 0, 1)]^\dagger={\bf 1}.
\end{eqnarray}
We also rewrite Eq.(\ref{TR}) as
\begin{eqnarray}\label{TR14}
\frac {1}{d}\sum _{\alpha, \beta } {\mathfrak D}(\alpha, \beta,0,0) \frac{\Theta}{{\rm Tr}(\Theta)} [{\mathfrak D}(\alpha, \beta,0,0)]^\dagger={\bf 1}.
\end{eqnarray}
We then add the last two equations and prove Eq.(\ref{TR55}).

\item[(4)]
For $\nu=0,1$ Eq.(\ref{121}) reduces to Eqs (\ref{AX1}),(\ref{AX2}) which have been given earlier.
We can also give a direct proof by calculating the matrix elements of both sides in the position basis, using Eq.(\ref{190}).

The last of Eqs(\ref{121})gives
\begin{eqnarray}
\frac{1}{d}\sum_{\alpha,\beta}{\mathfrak D}(\alpha, \beta,0,0)={\mathfrak P};\;\;\;\frac{1}{d}\sum_{\alpha,\beta}{\mathfrak D}(\alpha, \beta,0,1)={\bf 1}.
\end{eqnarray}
If we add and subtract these equations, we prove Eq.(\ref{70}).

\end{itemize}
\end{proof}

\subsection{$HWP(d)$ as solvable group and commutators that perform  displacements and parity transformations along loops in phase space}

\begin{lemma}
\mbox{}
\begin{itemize}
\item[(1)]
The commutator of two elements of $HWP(d)$ is
\begin{eqnarray}\label{123V}
{\cal L}_3^{(1)}(\alpha_1, \beta_1, \nu_1|\alpha_2, \beta_2, \nu_2)&=&[{\mathfrak D}(\alpha_1, \beta_1, \gamma_1, \nu_1),{\mathfrak D}(\alpha_2, \beta _2, \gamma_2,\nu_2)]=D(A,B,\Gamma)
\end{eqnarray}
Here
\begin{eqnarray}\label{WE1}
&&A=2{\cal A}(\alpha_1,\nu_1|\alpha_2,\nu_2);\;\;\;{\cal A}(\alpha_1,\nu_1|\alpha_2,\nu_2)=\alpha_1\nu_2-\alpha_2\nu_1\in{\mathbb Z}(d)\nonumber\\
&&{B}=2{\cal A}(\beta_1,\nu_1|\beta_2,\nu_2);\;\;\;{\cal A}(\beta_1,\nu_1|\beta_2,\nu_2)=\beta_1\nu_2-\beta_2\nu_1\in{\mathbb Z}(d),
\end{eqnarray}
and
\begin{eqnarray}\label{WE2}
&&\Gamma=\lambda(\nu_1,\nu_2){\cal A}(\alpha_1,\beta_1|\alpha_2,\beta_2)=\lambda(\nu_1,\nu_2)(\alpha_1\beta_2-\alpha_2\beta_1)\in{\mathbb Z}(d)\nonumber\\
&&\lambda(0,0)=1;\;\;\;\lambda(0,1)=\lambda(1,0)=\lambda(1,1)=-1.
\end{eqnarray}
${\cal L}_3^{(1)}(\alpha_1, \beta_1, \nu_1|\alpha_2, \beta_2, \nu_2)$ is an element of $HW(d)$ and does not depend on $\gamma_1, \gamma_2$.
In the special case the $\nu_1=\nu_2=0$, we get ${\cal L}_3^{(1)}(\alpha_1, \beta_1, 0|\alpha_2, \beta_2, 0)=\omega(\alpha_1\beta_2-\alpha_2\beta_1){\bf 1}$.
\item[(2)]
Since ${\cal L}_3^{(1)}(\alpha_1, \beta_1, \nu_1|\alpha_2, \beta_2, \nu_2)$ is an element of the non-Abelian group $HW(d)$, we also consider the `commutator of commutators' (Eq.(\ref{160})):
 \begin{eqnarray}\label{140}
 &&{\cal L}^{(2)}_3(\alpha_1, \beta_1, \nu_1;\alpha_2, \beta_2, \nu_2|\alpha_3, \beta_3, \nu_3;\alpha_4, \beta_4, \nu_4)\nonumber\\
&& =[{\cal L}_3^{(1)}(\alpha_1, \beta_1, \nu_1|\alpha_2, \beta_2, \nu_2),{\cal L}_3^{(1)}(\alpha_3, \beta_3, \nu_3|\alpha_4, \beta_4, \nu_4)]=\omega(\Phi){\bf 1},
\end{eqnarray}
where
\begin{eqnarray}
\Phi=4[{\cal A}(\alpha_1,\nu_1|\alpha_2,\nu_2){\cal A}(\beta_3,\nu_3|\beta_4,\nu_4)-{\cal A}(\alpha_3,\nu_3|\alpha_4,\nu_4){\cal A}(\beta_1,\nu_1|\beta_2,\nu_2)].
\end{eqnarray}
If $\nu_1=\nu_2=0$ or $\nu_3=\nu_4=0$, then $\Phi=0$. 
\end{itemize}
\end{lemma}
\begin{proof}
\mbox{}
\begin{itemize}
\item[(1)]
The commutator of two elements of $HWP(d)$ is
\begin{eqnarray}\label{123}
{\cal L}_3^{(1)}(\alpha_1, \beta_1, \nu_1|\alpha_2, \beta_2, \nu_2)&=&{\mathfrak D}(\alpha_1, \beta_1, \gamma_1, \nu_1){\mathfrak D}(\alpha_2, \beta _2, \gamma_2,\nu_2)\nonumber\\&\times&
[{\mathfrak D}(\alpha_1, \beta_1, \gamma_1, \nu_1)]^{-1}[{\mathfrak D}(\alpha_2, \beta _2, \gamma_2,\nu_2)]^{-1}.
\end{eqnarray}
The result for the product of the first two terms is given in Eq.(\ref{46}):
\begin{eqnarray}\label{46}
&&{\mathfrak D}(\alpha_1, \beta_1, \gamma_1, \nu_1){\mathfrak D}(\alpha_2, \beta _2, \gamma_2,\nu_2)={\mathfrak D}(A_1,B_1,\Gamma_1, \nu_1+\nu_2),
\end{eqnarray}
where
\begin{eqnarray}
&&A_1=\alpha_1+(-1)^{\nu_1}\alpha_2;\;\;\;B_1=\beta_1+(-1)^{\nu_1}\beta_2\nonumber\\
&&\Gamma_1= \gamma_1+\gamma_2+2^{-1}(-1)^{\nu_1}(\alpha_1\beta _2-\alpha_2\beta_1).
\end{eqnarray}
Also taking into account Eq.(\ref{inv}) for the inverse, we calculate the product of the first three terms in Eq.(\ref{123}):
\begin{eqnarray}
{\mathfrak D}(A_1,B_1,\Gamma_1, \nu_1+\nu_2){\mathfrak D}[(-1)^{\nu_1+1}\alpha_1, (-1)^{\nu_1+1}\beta_1, -\gamma_1, \nu_1]
={\mathfrak D}(A_2,B_2,\Gamma_2, \nu_2)
\end{eqnarray}
where
\begin{eqnarray}
&&A_2=A_1+(-1)^{\nu_1+\nu_2}(-1)^{\nu_1+1}\alpha_1=[1-(-1)^{\nu_2}]\alpha_1+(-1)^{\nu_1}\alpha_2=2\nu_2\alpha_1+(-1)^{\nu_1}\alpha_2\nonumber\\
&&B_2=B_1+(-1)^{\nu_1+\nu_2}(-1)^{\nu_1+1}\beta_1=[1-(-1)^{\nu_2}]\beta_1+(-1)^{\nu_1}\beta_2=2\nu_2\beta_1+(-1)^{\nu_1}\beta_2.
\end{eqnarray}
We used here Eq(\ref{40}) which is valid for $\nu=0,1$. Also
\begin{eqnarray}
\Gamma_2&=& \gamma_2+2^{-1}(-1)^{\nu_1}(\alpha_1\beta _2-\alpha_2\beta_1)+2^{-1}(-1)^{\nu_1+\nu_2}[A_1(-1)^{\nu_1+1}\beta_1-B_1(-1)^{\nu_1+1}\alpha_1]\nonumber\\
&=&\gamma_2+2^{-1}(-1)^{\nu_1}[1+(-1)^{\nu_2}](\alpha_1\beta _2-\alpha_2\beta_1).
\end{eqnarray}
Then we show that the product of all four terms in Eq.(\ref{123}) is:
\begin{eqnarray}
{\mathfrak D}(A_2,B_2,\Gamma_2, \nu_2){\mathfrak D}[(-1)^{\nu_2+1}\alpha_2, (-1)^{\nu_2+1}\beta_2, -\gamma_2, \nu_2]={\mathfrak D}(A,B,\Gamma, 0)
\end{eqnarray}
where
\begin{eqnarray}
&&A=A_2+(-1)^{\nu_2}(-1)^{\nu_2+1}\alpha_2=A_2-\alpha_2=2\nu_2\alpha_1-[1-(-1)^{\nu_1}]\alpha_2=2\nu_2\alpha_1-2\nu_1\alpha_2\nonumber\\
&&B=B_2+(-1)^{\nu_2}(-1)^{\nu_2+1}\beta_2=B_2-\beta_2=2\nu_2\beta_1-[1-(-1)^{\nu_1}]\beta_2=2\nu_2\beta_1-2\nu_1\beta_2.
\end{eqnarray}
This confirms Eq.(\ref{WE1}). Also
\begin{eqnarray}
\Gamma&=&\Gamma_2- \gamma_2+2^{-1}(-1)^{\nu_2}[A_2(-1)^{\nu_2+1}\beta _2-B_2(-1)^{\nu_2+1}\alpha _2].
\end{eqnarray}
But
\begin{eqnarray}
A_2(-1)^{\nu_2+1}\beta _2-B_2(-1)^{\nu_2+1}\alpha _2=2\nu_2(-1)^{\nu_2+1}(\alpha_1\beta_2-\alpha_2\beta_1),
\end{eqnarray}
and therefore
\begin{eqnarray}
\Gamma=\left \{2^{-1}(-1)^{\nu_1}[1+(-1)^{\nu_2}]-\nu_2\right\}(\alpha_1\beta _2-\alpha_2\beta_1).
\end{eqnarray}
We easily check that the prefactor takes the values given in Eq.(\ref{WE2}).

We note that $A$ is twice the area which is the magnitude of the exterior product of the vectors $(\alpha_1,\nu_1), (\alpha_2,\nu_2)$ in the phase space $PS(Z,{\mathfrak P})$.
$B$  is twice the area which is the magnitude of exterior product of  the vectors $(\beta_1,\nu_1), (\beta_2,\nu_2)$ in the phase space $PS(X,{\mathfrak P})$.
$\Gamma$ is the area which is the magnitude of the exterior product of  the vectors $(\alpha_1,\beta_1), (\alpha_2, \beta_2)$ in the phase space $PS(X,Z)$, times $\lambda(\nu_1,\nu_2)$.

In the special case the $\nu_1=\nu_2=0$, we easily find that ${\cal L}_3^{(1)}(\alpha_1, \beta_1, 0|\alpha_2, \beta_2, 0)=\omega(\alpha_1\beta_2-\alpha_2\beta_1){\bf 1}$.
\item[(2)]
This is proved using Eqs.(\ref{123V}), (\ref{199n}).
In the special case that $\nu_1=\nu_2=0$ we get ${\cal A}(\alpha_1,0|\alpha_2,0)={\cal A}(\beta_1,0|\beta_2,0)=0$ and therefore $\Phi=0$.
We get similar result when $\nu_3=\nu_4=0$. 

\end{itemize}
\end{proof}

${\cal L}_3^{(1)}(\alpha_1, \beta_1, \nu_1|\alpha_2, \beta_2, \nu_2)$ describes transport of quantum states along a loop defined by $(\alpha_1, \beta_1, \nu_1)$ and $(\alpha_2, \beta_2, \nu_2)$ in 
${\mathbb Z}(d)\times {\mathbb Z}(d)\times{\mathbb Z}(2)$
(which we denote as $PS(Z,X,{\mathfrak P})$). 
Special cases of ${\cal L}_3^{(1)}(\alpha_1, \beta_1, \nu_1|\alpha_2, \beta_2, \nu_2)$ are:
\begin{eqnarray}
&&{\cal L}_3^{(1)}(\alpha_1, \beta_1, 0|\alpha_2, \beta_2, 0)={\cal L}_1(\alpha_1, \beta_1|\alpha_2, \beta_2)\nonumber\\
&&{\cal L}_3^{(1)}(\alpha_1, 0, \nu_1|\alpha_2, 0, \nu_2)={\cal L}_2(\alpha_1,  \nu_1|\alpha_2,  \nu_2).
\end{eqnarray}

${\cal L}^{(2)}_3(\alpha_1, \beta_1, \nu_1;\alpha_2, \beta_2, \nu_2|\alpha_3, \beta_3, \nu_3;\alpha_4, \beta_4, \nu_4)$ describes transport in a `loop of four loops' in $PS(Z,X,{\mathfrak P})$.
The first loop is  defined by $(\alpha_1, \beta_1, \nu_1)$ and $(\alpha_2, \beta_2, \nu_2)$, the second loop by $(\alpha_3, \beta_3, \nu_3)$ and $(\alpha_4, \beta_4, \nu_4)$, the third loop
is the first loop in the opposite direction, and the fourth loop is the second loop in the opposite direction. 

Both ${\cal L}_3^{(1)}$ and ${\cal L}_3^{(2)}$ are products of displacement operators, and therefore they are displacement operators. 
${\cal L}_3^{(1)}$ belongs in $HW(d)$ and ${\cal L}_3^{(2)}$ in ${\cal G}_D({\bf 1})\cong{\mathbb Z}(d)$.
Physically, each displacement operator can be implemented stroboscopically as a time evolution operator
with Hamiltonians that are the principal logarithms of the displacement operators as in Eqs.(\ref{10A}),(\ref{10B}).
In this sense ${\cal L}_3^{(1)}$ and ${\cal L}_3^{(2)}$ are time evolution operators (and their principal logarithms are Hermitian operators).

\begin{proposition}\label{pro45}
\mbox{}
\begin{itemize}
\item[(1)]
\begin{eqnarray}\label{77A}
[HWP(d),HWP(d)]=HW(d);\;\;\;[HW(d),HW(d)]={\cal G}_d({\bf 1});\;\;\;[{\cal G}_d({\bf 1}), {\cal G}_d({\bf 1})]\cong \{{\bf 1}\}.
\end{eqnarray}
\item[(2)]
$HWP(d)$ is a solvable group (of solvability class $3$) with derived series
\begin{eqnarray}\label{78A}
HWP(d)\triangleright HW(d) \triangleright {\cal G}_d({\bf 1})\triangleright \{{\bf 1}\}.
\end{eqnarray}

\item[(3)]
Here the Abelian groups in Eq.(\ref{100}) are
\begin{eqnarray}\label{15AA}
HWP(d)/HW(d)\cong {\mathbb Z}(2);\;\;\;HW(d)/{\cal G}_d({\bf 1})\cong {\mathbb Z}( d)\times {\mathbb Z}( d);\;\;\;{\cal G}_d({\bf 1})/\{{\bf 1}\}\cong {\cal G}_d({\bf 1}).
\end{eqnarray}
Physically ${\mathbb Z}(2)$ is related to parity and ${\mathbb Z}(d)\times {\mathbb Z}( d)=PS(X,Z)$ is the position-momentum phase space of the quantum system.
\item[(4)]
$HWP(d)$ is not a nilpotent group.
\end{itemize}

\end{proposition}

\begin{proof}
\mbox{}
\begin{itemize}
\item[(1)]
The commutators ${\cal L}_3^{(1)}(\alpha_1, \beta_1, \nu_1|\alpha_2, \beta_2, \nu_2)$ generate $HW(d)$, which is the commutator subgroup (or derived subgroup) of $HWP(d)$. The rest of the commutator groups are the same as in Eq.(\ref{13}).

\item[(2)]
From Eq.(\ref{77A}), follows immediately that $HWP(d)$ is a solvable group and  its derived series is the one in Eq.(\ref{78A}).
\item[(3)]
The fact that $HWP(d)$ is the semidirect product $HW(d)\rtimes {\mathbb Z}(2)$ implies that $HWP(d)/HW(d)\cong {\mathbb Z}(2)$.
The rest of the relations in Eq.(\ref{15AA}) are the same as in Eq.(\ref{150}).
\item[(4)]
We first prove that the commutator of $HWP(d)$ with $HW(d)$ is
\begin{eqnarray}\label{73}
&&[HWP(d),HW(d)]=HW(d).
\end{eqnarray}
We consider the commutator in Eq.(\ref{123}) with $\nu_2=0$ (so that ${\mathfrak D}(\alpha_2, \beta _2, \gamma_2,0) \in HW(d)$) and we easily see that the result is an element of $HW(d)$.
This proves Eq.(\ref{73}). 

Therefore the lower central series
\begin{eqnarray}
HWP(d);\;\;\;[HWP(d),HW(d)]=HW(d);\;\;\;[HWP(d), [HWP(d),HW(d)]]=HW(d)
\end{eqnarray}
terminates in $HW(d)$, and not in $\{\bf 1\}$. Consequently $HWP(d)$ is not a nilpotent group.
\end{itemize}
\end{proof}

 The principal value of the logarithm of ${\cal L}_3^{(1)}(\alpha_1, \beta_1|\alpha_2, \beta_2)$  is:
 \begin{eqnarray}\label{19AA}
\frac{d}{2\pi i}\log [{\cal L}_3^{(1)}(\alpha_1, \beta_1, \nu_1|\alpha_2, \beta_2, \nu_2)]=\Psi(\alpha_1, \beta_1, \nu_1|\alpha_2, \beta_2, \nu_2)=\sum _j\psi _j\ket{\ell_j}\bra{\ell_j}
\end{eqnarray} 
Here $\psi_j$ and $\ket{\ell_j}$ are the (real) eigenvalues and eigenvectors of $\Psi$. Then
 \begin{eqnarray}
{\cal L}_3^{(1)}(\alpha_1, \beta_1, \nu_1|\alpha_2, \beta_2, \nu_2)=\sum _j\omega(\psi_j)\ket{\ell_j}\bra{\ell_j}.
\end{eqnarray} 
Acting with ${\cal L}_3^{(1)}(\alpha_1, \beta_1, \nu_1|\alpha_2, \beta_2, \nu_2)\in HW(d)$ on a state $\ket{f}$, we get
\begin{eqnarray}
{\cal L}_3^{(1)}(\alpha_1, \beta_1, \nu_1|\alpha_2, \beta_2, \nu_2)\ket{f}=\sum _j\omega(\psi_j) f_j\ket{\ell_j};\;\;\;f_j=\bra{\ell_j}f\rangle.
\end{eqnarray} 
We note that
\begin{eqnarray}\label{loop1}
|\bra{f}{\cal L}_3^{(1)}(\alpha_1, \beta_1, \nu_1|\alpha_2, \beta_2, \nu_2)\ket{f}|=\left |\sum _j\omega(\psi_j)|f_j|^2\right |\le 1.
\end{eqnarray}
When the commutator ${\cal L}_3^{(1)}(\alpha_1, \beta_1, \nu_1|\alpha_2, \beta_2, \nu_2)$ acts on quantum state, the new state has a different direction from the original one.
This should be compared with Eq.(\ref{loop}).

\begin{example}
In $HWP(3)$ we multiplied numerically the $4$ matrices in Eq.(\ref{123V}) and we found
\begin{eqnarray}
{\cal L}_3^{(1)}(2,1,1|1,2,0)=\begin{pmatrix}
   0 & -0.50 +0.86i   &0\nonumber\\
   0  &0   &-0.50-0.86i\nonumber\\
  1 & 0   &0
\end{pmatrix}.
\end{eqnarray}
In this case Eq.(\ref{123V}) gives
\begin{eqnarray}
{\cal L}_3^{(1)}(2,1,1|1,2,0)={\mathfrak D}(1,2,0,0)=D(1,2,0),
\end{eqnarray}
and we confirmed numerically that this is the matrix given above.

We also calculated numerically the corresponding Hermitian matrix in Eq.(\ref{19AA}):
\begin{eqnarray}
\Psi(2,1,1|1,2,0)=
\begin{pmatrix}
  0  &0.50 + 0.28i   &0.57i\nonumber\\
  0.50 - 0.28i  &0 &- 0.50+0.28i\nonumber\\
   -0.57i  &-0.50-0.28i  & 0
\end{pmatrix}.
\end{eqnarray} 
Acting with ${\cal L}_3^{(1)}(2,1,1|1,2,0)$ on the state $\ket {f}$ given below, we get
\begin{eqnarray}
|\bra{f}{\cal L}_3^{(1)}(2,1,1|1,2,0)\ket{f}|=0.184;\;\;\;
\ket{f}=\begin{pmatrix}
0.3\\
0.4i\\
0.87
\end{pmatrix}.
\end{eqnarray}
This should be compared with Eq.(\ref{loop}).

We also calculated ${\cal L}_3^{(2)}(2,1,1;1,2,0|1,0,1;2,2,0)$. In this example
\begin{eqnarray}
{\cal A}(\alpha_1,\nu_1|\alpha_2,\nu_2)=-1;\;\;\;
{\cal A}(\beta_3,\nu_3|\beta_4,\nu_4)=-2;\;\;\;
{\cal A}(\alpha_3,\nu_3|\alpha_4,\nu_4)=-2;\;\;\;
{\cal A}(\beta_1,\nu_1|\beta_2,\nu_2)=-2,
\end{eqnarray}
modulo $3$. Therefore
\begin{eqnarray}
{\cal L}_3^{(2)}(2,1,1;1,2,0|1,0,1;2,2,0)=\omega(-2){\bf 1}.
\end{eqnarray}
$\omega(\alpha)$ is given in Eq.(\ref{FF}) with $d=3$.
\end{example}
\begin{example}
We give a similar example in $HWP(5)$. Using Eq.(\ref{123V}) we found
\begin{eqnarray}
{\cal L}_3^{(1)}(2,3,1|1,4,0)=D(3,2,0)
\end{eqnarray}
We also calculated ${\cal L}_3^{(2)}(2,3,1;1,2,0|1,0,1;3,3,0)$ using Eq.(\ref{140}).
In this example
\begin{eqnarray}
{\cal A}(\alpha_1,\nu_1|\alpha_2,\nu_2)=-1;\;\;\;
{\cal A}(\beta_3,\nu_3|\beta_4,\nu_4)=-3;\;\;\;
{\cal A}(\alpha_3,\nu_3|\alpha_4,\nu_4)=-3;\;\;\;
{\cal A}(\beta_1,\nu_1|\beta_2,\nu_2)=-2,
\end{eqnarray}
modulo $5$. Therefore 
\begin{eqnarray}
{\cal L}_3^{(2)}(2,3,1;1,2,0|1,0,1;3,3,0)=\omega(3){\bf 1}.
\end{eqnarray}
$\omega(\alpha)$ is given in Eq.(\ref{FF}) with $d=5$.
\end{example}

\subsection{Coherent states related to the $HWP(d)$ group}\label{sec35}

We extend the set ${\cal C}[HW(d)]$ of coherent states as follows.
We act with the ${\mathfrak D}(\alpha, \beta, 0, \nu)$ on a fiducial state $\ket{s}$ and we get the $2d^2$ coherent states:
\begin{eqnarray}\label{TR79}
\ket{C; \alpha, \beta, \nu}={\mathfrak D}(\alpha, \beta, 0, \nu)\ket{s};\;\;\;\ket{s}=\sum _m s_m\ket{X;r};\;\;\;\sum _m|s_m|^2=1.
\end{eqnarray}
The $C$ in the notation is not a variable, but it indicates coherent states. 
We denote as ${\cal C}[HWP(d)]$ the set of these coherent states. 

We have explained earlier that the fiducial vector $\ket{s}$, should be a generic vector so that the coherent states are different from each other.
In addition to the requirement that $\ket{s}$ should not be a position or momentum state, here we also require that $\ket{s}$ should not be an eigenstate of ${\mathfrak P}$ 
(${\mathfrak P}\ket{s}\ne \pm \ket{s}$), so that
$\ket{C; \alpha, \beta, 1}\ne \pm \ket{C; \alpha, \beta, 0}$.
This is equivalent to $\varpi _0\ket{s}\ne 0$ and $\varpi _1\ket{s}\ne 0$.

It is easily seen that the set of coherent states ${\cal C}[HW(d)]$ introduced in section \ref{sec24}, is a subset of ${\cal C}[HWP(d)]$.
Furthermore, we show below that the new coherent states in ${\cal C}[HWP(d)]$ (i.e., the coherent states in the set difference ${\cal C}[HWP(d)]\setminus {\cal C}[HW(d)]$) are Fourier transforms of those in ${\cal C}[HW(d)]$.

\begin{proposition}\label{pro46}
The following are properties of the coherent states $\ket{C; \alpha, \beta, \nu}$:
\begin{itemize}
\item[(1)]
Closure under displacement transformations in $HWP(d)$: acting with displacement transformations  in $HWP(d)$ on the coherent states in ${\cal C}[HWP(d)]$, we get other coherent states in ${\cal C}[HWP(d)]$.
\item[(2)]
The resolution of the identity
\begin{eqnarray}\label{TR79}
\frac {1}{2d}\sum_{\nu=0}^1\sum _{\alpha, \beta } \ket{C; \alpha, \beta, \nu} \bra{C;\alpha, \beta, \nu}={\bf 1}.
\end{eqnarray}

\item[(3)]
The coherent states $\ket{C; \alpha, \beta, 0}$ are the same as the coherent states $\ket{C; \alpha, \beta}$ in Eq.(\ref{coh1}) which are related to $HW(d)$ (which is a subgroup of $HWP(d)$).
The coherent states $\ket{C; \alpha, \beta, 1}$ are related to $\ket{C; \alpha, \beta, 0}$ through the Fourier transform
\begin{eqnarray}\label{WW22}
\frac {1}{d}\sum _{\alpha ,\beta } \omega[2^{-1}(-1)^\nu(\beta \gamma -\alpha \delta)]\ket{C;\alpha, \beta,\nu} =\ket{C;\gamma, \delta, \nu+1};\;\;\;\nu\in{\mathbb Z}(2).
\end{eqnarray}
\item[(4)]
We can expand an arbitrary (normalised) state $\ket{f}$ in $H(d)$, in terms of the $2d^2$ coherent states:
\begin{eqnarray}\label{A13}
&&\ket{f}=\frac{1}{2d}\sum_{\nu=0}^1\sum _{\alpha, \beta } F(\alpha, \beta, \nu)\ket{C; \alpha, \beta, \nu};\;\;\;F(\alpha, \beta, \nu)=\bra{C;\alpha, \beta, \nu} f \rangle
\end{eqnarray}
The $F(\alpha, \beta, \nu)$ are Bargmann coefficients with respect to the coherent stets $\ket{C; \alpha, \beta, \nu}$ (but there is no analyticity in the present discrete context ).
The scalar product of two states is given in terms of them as 
\begin{eqnarray}\label{A13D}
&&\langle g\ket{f}=\frac{1}{2d}\sum_{\nu=0}^1\sum _{\alpha, \beta } [G(\alpha, \beta, \nu)]^*F(\alpha, \beta, \nu).
\end{eqnarray}
An arbitrary state $\ket{f}$ can be expanded in terms of the $d^2$ coherent states in  ${\cal C}[HW(d)]$ (Eq.(\ref{A12})), and it can also be 
expanded in terms of the $2d^2$ coherent states in  ${\cal C}[HWP(d)]$ (Eq.(\ref{A13})).
In an example below, we show that the latter is more immune in the presence of noise.
\item[(5)]
The Bargmann coefficients $F(\alpha, \beta, 0)$ are the same as the  $f(\alpha, \beta)$ in Eq.(\ref{A12}).
The $F(\alpha, \beta, 0)$ and $F(\alpha, \beta, 1)$ are related through the Fourier transform
\begin{eqnarray}\label{A24}
\frac {1}{d}\sum _{\alpha ,\beta } \omega[-2^{-1}(-1)^\nu(\beta \gamma -\alpha \delta)]F(\alpha, \beta,\nu) =F(\gamma, \delta, \nu+1);\;\;\;\nu\in{\mathbb Z}(2).
\end{eqnarray}

\item[(6 )]
The
\begin{eqnarray}\label{A12F}
Q(\alpha, \beta,\nu )=|F(\alpha, \beta, \nu)|^2;\;\;\;\frac{1}{2d}\sum _{\nu=0}^1\sum _{\alpha, \beta } Q(\alpha, \beta, \nu)=1,
\end{eqnarray}
can be viewed as an extended $Q$-function (extended Husimi function) in the present context (and depends on the fiducial vector $\ket{s}$). 
The $Q(\alpha, \beta, 0)$ are the same as the $Q(\alpha, \beta)$ in Eq.(\ref{A12C}).
The $\frac{1}{2d}Q(\alpha, \beta, \nu)$ can be interpreted as a pseudo-probability distribution.

\item[(7)]
Non-orthogonality relation:
\begin{eqnarray}
&&|\langle C;\alpha_1, \beta_1, \nu_1\ket{C;\alpha _2, \beta_2, \nu_2}|^2=\left |\sum _m \omega [(-1)^{\nu_2+1}\alpha_1m+(-1)^{\nu_2}\alpha_2m]s_k^*s_m\right|^2\nonumber\\
&&k=(-1)^{\nu_1+\nu_2}m+(-1)^{\nu_1+1}\beta_1+(-1)^{\nu_1}\beta_2.
\end{eqnarray}

\end{itemize}
\end{proposition}
\begin{proof}
\mbox{}
\begin{itemize}
\item[(1)]
Using Eq.(\ref{46}) we prove that:
\begin{eqnarray}
&&{\mathfrak D}(\alpha_1, \beta_1, \gamma_1, \nu_1)\ket{C; \alpha_2, \beta_2, \nu_2}= \omega(\Gamma)\ket{C;A,B, \nu_1+\nu_2}\nonumber\\
&&A=\alpha_1+(-1)^{\nu_1}\alpha_2;\;\;\;B=\beta_1+(-1)^{\nu_1}\beta_2;\;\;\;\nu_1,\nu_2=0,1\nonumber\\
&&\Gamma= \gamma_1+2^{-1}(-1)^{\nu_1}(\alpha_1\beta _2-\alpha_2\beta_1).
\end{eqnarray}
\item[(2)]
The resolution of the identity follows from Eq.(\ref{TR55}) with $\Theta=\ket{s}\bra{s}$.

\item[(3)]
Using Eq.(\ref{WW2}) we prove Eq.(\ref{WW22}).

\item[(4)]
Using the resolution of the identity we prove Eq.(\ref{A13}) and also Eq.(\ref{A13D}) .

\item[(5)]
The Fourier transform in Eq.(\ref{A24}) follows immediately from the Fourier transform in Eq.(\ref{WW22}).
\item[(6)]
The first part of Eq.(\ref{A12F}) is definition, and the second part follows from Eq.(\ref{A13D}).

\item[(7)]
This is proved using Eq.(\ref{46}) for the product of two displacement operators, and then Eq.(\ref{190}).

\end{itemize}

\end{proof}

Further work can study SIC-POVM in the present case with $2d^2$ coherent states (existing work is mainly for $d^2$ POVM).
The problem here is to find fiducial vectors such that all the $|\langle C;\alpha_1, \beta_1, \nu_1\ket{C;\alpha _2, \beta_2, \nu_2}|^2$ 
with $(\alpha_1, \beta_1, \nu_1)\ne (\alpha _2, \beta_2, \nu_2)$, are equal to each other. In this case using the resolution of the identity we show that
\begin{eqnarray}\label{SIC}
&&|\langle C;\alpha_1, \beta_1, \nu_1\ket{C;\alpha _2, \beta_2, \nu_2}|^2=\frac{2d-1}{2d^2-1};\;\;\;(\alpha_1, \beta_1, \nu_1)\ne (\alpha _2, \beta_2, \nu_2).
\end{eqnarray}

\section{The use of expansions in terms of coherent states in noisy problems}\label{sec379}
\begin{table} 
\caption{The Bargmann coefficients $F(\alpha, \beta,\nu)$ and the $Q$-function $Q(\alpha, \beta,\nu)$ (Eqs.(\ref{A13}), (\ref{A12F})) for the state $\ket{f}$ 
and the fiducial vector $\ket{s}$ in Eq.(\ref{state}). The $F(\alpha, \beta, 0)$ and $Q(\alpha, \beta, 0)$ are also the $f(\alpha, \beta)$ and $Q(\alpha, \beta)$ correspondingly (Eqs.(\ref{A12}), (\ref{A12C})).
The unified Wigner-Weyl function ${\mathfrak W}( \Theta; \alpha, \beta, \nu)$ is also given for $\Theta=\ket{f}\bra{f}$ (example \ref{ex7})}
\begin{tabular}{|c|c|c|c|c|c|}\hline
$\nu$&$\alpha$&$\beta$&$F(\alpha, \beta, \nu)$&$Q(\alpha, \beta,\nu)$&${\mathfrak W}( \Theta; \alpha, \beta, \nu)$\\\hline
0&-1&-1&0.06+0.04i&0&-0.27+0.71i\\\hline
0&-1&0&0.05+0.13i&0.02&-0.35+0.06i\\\hline
0&-1&1&-0.29+0.79i&0.72&0.08-0.38i\\\hline
0&0&-1&0.83&0.69&0.19+0.23i\\\hline
0&0&0&0.49+0.23i&0.29&0.98\\\hline
0&0&1&0.15+0.28i&0.10&0.19-0.23i\\\hline
0&1&-1&0.06-0.04i&0.01&0.08+0.38i\\\hline
0&1&0&-0.54-0.72i&0.82&-0.35-0.06i\\\hline
0&1&1&0.14+0.53i&0.31&-.027-0.71i\\\hline
1&-1&-1&0.27+0.11i&0.08&0.04\\\hline
1&-1&0&0.30-0.11i&0.10&0.86\\\hline
1&-1&1&0.33+0.67i&0.57&0.29\\\hline
1&0&-1&0.92&0.85&0.41\\\hline
1&0&0&0.32+0.41i&0.27&0.10\\\hline
1&0&1&0.23+0.09i&0.06&0.87\\\hline
1&1&-1&0.27-0.11i&0.09&0.77\\\hline
1&1&0&-0.62-0.66i&0.83&-0.68\\\hline
1&1&1&0.10+0.27i&0.09&0.29\\\hline
\end{tabular} \label{t1}
\end{table}

Redundancy is very important in noisy problems. For example, without redundancy in our language (quantified by Shannon and later by many others) we would not be able to communicate.
A small spelling mistake would change completely the meaning of a sentence. Due to redundancy, the meaning of a sentence does not change in the presence of small spelling mistakes.

In our case redundancy is the fact that the number of coherent states is greater that the dimension $d$ ($d^2$ coherent states in section \ref{sec24}, and $2d^2$ coherent states in section \ref{sec35}).
We demonstrate with examples, that the expansion in Eq.(\ref{A13}) that involves more coherent states is more immune to noise, than the expansion in Eq.(\ref{A12}) that involves fewer coherent states.

\subsection{Example in $H(3)$}

In $H(3)$ we consider the $9$ coherent states in ${\cal C}[HW(3)]$, and also the $18$ coherent states in ${\cal C}[HWP(3)]$. The fiducial vector in both cases is $\ket{s}$.
We also consider the state $\ket{f}$, where (in the basis in Eq.(\ref{30AB})):
\begin{eqnarray}\label{state}
\ket{s}=\begin{pmatrix}
0.5\\0.4i\\0.77
\end{pmatrix};\;\;\;
\ket{f}=\begin{pmatrix}
0.7i\\0.3\\0.64
\end{pmatrix}.
\end{eqnarray}

\subsubsection{Expansion in terms of coherent states in ${\cal C}[HW(3)]$}

We calculate numerically the Bargmann coefficients $f(\alpha, \beta)=\bra{C;\alpha, \beta} f \rangle$ (Eq.(\ref{A12})), and the $Q$-function (Eq(\ref{A12C})).
The $f(\alpha, \beta)$ and $Q(\alpha, \beta)$ are presented as $F(\alpha, \beta,0)$ and $Q(\alpha, \beta,0)$ correspondingly, in table \ref{t1} (top half of this table).
The results depend on the fiducial vector.
It is easily checked that the constraint in Eq.(\ref{A12C}) holds.

We then add noise $\epsilon (\alpha, \beta)$ to the Bargmann coefficients, which are random numbers uniformly distributed in some interval $[-E,E]$ (produced by MATLAB). 
We calculate numerically the `noisy state'
\begin{eqnarray}\label{E1}
&&\ket{f_1}=\frac{1}{d}\sum _{\alpha, \beta } f_1(\alpha, \beta)\ket{C; \alpha, \beta};\;\;\;f_1(\alpha, \beta)=f(\alpha, \beta)+\epsilon(\alpha, \beta).
\end{eqnarray}
We define the error as
\begin{eqnarray}\label{E2}
&&{\cal E}_1=|\langle f_1\ket{f}-1|=\left |\frac{1}{d}\sum _{\alpha, \beta } \epsilon (\alpha, \beta)f(\alpha, \beta)\right |.
\end{eqnarray}
Since there are random numbers in the above formula, the result is slightly different every time we run the programme. 
For $E=0.1$ (random numbers in the interval $[-0.1, 0.1]$) we run the programme $5$ times and we found ${\cal E}_1$ to be
$0.016$, $0.036$, $0.051$, $0.011$, $0.033$ (average value $0.030$). 

${\cal E}_1$ depends weakly on the fiducial vector.
In order to show this we repeated the calculation for the same state $\ket{f}$ and with the fiducial vectors
\begin{eqnarray}
\ket{s_1}=\begin{pmatrix}
0.6\\-0.5\\0.66
\end{pmatrix};\;\;\;
\ket{s_2}=\begin{pmatrix}
0.1\\-0.5\\0.86
\end{pmatrix}.
\end{eqnarray}
We run the programme $5$ times with the fiducial vector $\ket{s_1}$ we found ${\cal E}_1$ to be
$0.013$, $0.021$, $0.031$, $0.061$, $0.004$ (average value $0.026$). 
With the fiducial vector $\ket{s_2}$ we found ${\cal E}_1$ to be
$0.014$, $0.019$, $0.048$, $0.001$, $0.033$ (average value $0.023$). 
So the dependence of ${\cal E}_1$ on the fiducial vector, is weak.

\subsubsection{Expansion in terms of coherent states in ${\cal C}[HWP(3)]$}

We calculate numerically the $18$ Bargmann coefficients $F(\alpha, \beta, \nu)=\bra{C;\alpha, \beta, \nu} f \rangle$ (Eq.(\ref{A13})) for the same state $\ket{f}$ and with the same fiducial vector $\ket{s}$
(given in Eq.(\ref{state})).
Numerical results for $F(\alpha, \beta, \nu)$ and also for the $Q$-function $Q(\alpha, \beta, \nu)$ are given in table \ref{t1}.
We can check that the constraint in Eq.(\ref{A12F}) holds.

We then add  noise $\epsilon (\alpha, \beta, \nu)$ to the Bargmann  coefficients, which are random numbers uniformly distributed in the same interval $[-0.1,0.1]$. 
We calculate numerically the `noisy state'
\begin{eqnarray}\label{E3}
&&\ket{f_2}=\frac{1}{2d}\sum _{\alpha, \beta } f_2(\alpha, \beta, \nu)\ket{C; \alpha, \beta, \nu};\;\;\;f_2(\alpha, \beta, \nu)=f(\alpha, \beta, \nu)+\epsilon(\alpha, \beta, \nu),
\end{eqnarray}
and  the error 
\begin{eqnarray}\label{E4}
&&{\cal E}_2=|\langle f_2\ket{f}-1|=\left |\frac{1}{2d}\sum_{\nu=0}^1\sum _{\alpha, \beta } \epsilon (\alpha, \beta, \nu)f(\alpha, \beta, \nu)\right |.
\end{eqnarray}
As above we run the programme $5$ times. We found ${\cal E}_2$ to be
$0.004$, $0.028$, $0.036$, $0.011$, $0.016$ (average value $0.019$). This is smaller than ${\cal E}_1$ in the expansion with $d^2$ coherent states.

\subsection{Example in $H(5)$}
We repeat the calculation described above, with an example in $H(5)$.
Here the fiducial vector $\ket{s}$ and the state $\ket {f}$ are
\begin{eqnarray}\label{fid5}
\ket{s}=\begin{pmatrix}
0.5\\0.4\\0.3i\\0.6\\0.37
\end{pmatrix};\;\;\;
\ket{f}=\begin{pmatrix}
0.65\\0.3\\-0.3\\0.5\\0.38i
\end{pmatrix}.
\end{eqnarray}
We calculated numerically the $25$ Bargmann coefficients $f(\alpha, \beta)=\bra{C;\alpha, \beta} f \rangle$.
We then added noise $\epsilon (\alpha, \beta)$ to the Bargmann coefficients, which are random numbers uniformly distributed in $[-0.1,0.1]$. 
We calculated numerically the `noisy state' in Eq.(\ref{E1}) and the error in Eq.(\ref{E2}).
We run the programme $5$ times and we found ${\cal E}_1$ to be
$0.0077$, $0.0135$, $0.0137$, $0.0288$, $0.0235$ (average value $0.0174$). 

We also calculated numerically the $50$ Bargmann coefficients $F(\alpha, \beta, \nu)=\bra{C;\alpha, \beta, \nu} f \rangle$ for the same state $\ket{f}$.
We then add  noise $\epsilon (\alpha, \beta, \nu)$ to the Bargmann  coefficients, which are random numbers uniformly distributed in the same interval $[-0.1,0.1]$. 
We calculate numerically the `noisy state' in Eq.(\ref{E3}) and the error in Eq.(\ref{E4}).
As above we run the programme $5$ times. We found ${\cal E}_2$ to be
$0.0059$, $0.0040$, $0.0095$, $0.0252$, $0.0124$ (average value $0.0114$). This is smaller than ${\cal E}_1$ in the expansion with $d^2$ coherent states. 

Overall in problems with noise, the expansion in Eq.(\ref{A13}) with $2d^2$ coherent states  seems to be more robust than the expansion in Eq.(\ref{A12}) with $d^2$ coherent states.

\section{The unified Wigner-Weyl function}\label{sec18}

Let $\Theta$ be an operator.
We define its combined Wigner-Weyl function in terms of the $2d^2$ operators ${\mathfrak D}(\alpha, \beta, 0, \nu)$ (which are elements of $HWP(d)/{\mathbb Z}(d)$ group):
\begin{eqnarray}
{\mathfrak W}(\Theta; \alpha, \beta, \nu)={\rm Tr}[\Theta {\mathfrak D}(\alpha, \beta, 0, \nu)].
\end{eqnarray}
Special cases of this are the  Wigner and Weyl functions introduced earlier
\begin{eqnarray}
&&{\mathfrak W}(\Theta; \alpha, \beta, 1)=W(\Theta; 2^{-1}\alpha, 2^{-1}\beta);\;\;\;{\mathfrak W}(\Theta; \alpha, \beta, 0)={\widetilde W}(\Theta; \alpha, \beta)\nonumber\\
&&{\mathfrak W}(\Theta; \alpha, 0, \nu)={\mathfrak W}_Z(\Theta; \alpha, \nu);\;\;\;{\mathfrak W}(\Theta; 0, \beta, \nu)={\mathfrak W}_X(\Theta;  \beta, \nu)
\end{eqnarray}

Using proposition \ref{pro1} we get the following:
\begin{proposition}\label{pro7}
\mbox{}
\begin{itemize}
\item[(1)]
The Wigner-Weyl function ${\mathfrak W}(\Theta; \alpha, \beta,\nu) $ is related to ${\mathfrak W}(\Theta; \gamma, \delta, \nu+1)$ through a Fourier transform:
\begin{eqnarray}\label{WW3}
\frac {1}{d}\sum _{\alpha ,\beta } {\mathfrak W}(\Theta; \alpha, \beta,\nu) \omega[2^{-1}(-1)^\nu(\beta \gamma -\alpha \delta)]={\mathfrak W}(\Theta; \gamma, \delta, \nu+1).
\end{eqnarray}
This combines both Eqs.(\ref{QW10}),(\ref{QW11}).
\item[(2)]
An operator $\Theta$ can be expanded in terms of the $2d^2$ operators ${\mathfrak D}(\alpha, \beta, 0, \nu)$ with the Wigner-Weyl functions as coefficients:
\begin{eqnarray}\label{300}
\Theta=\frac{1}{2d}\sum_{\nu=0}^1\sum _{\alpha ,\beta} {\mathfrak W}[\Theta; (-1)^{\nu+1} \alpha ,(-1)^{\nu+1} \beta ,\nu) {\mathfrak D}(\alpha, \beta, 0, \nu)
\end{eqnarray}
There is redundancy in this expansion, in the sense that the operator $\Theta$  is a $d\times d$ matrix with $d^2$ complex elements, and is represented with the $2d^2$ complex numbers
${\mathfrak W}(\Theta;\alpha ,\beta ,\nu) $ (in contrast, there is no redundancy in either of the two expansions in Eq.(\ref{exp})). 
Due to this redundancy the expansion of $\Theta$ in terms of ${\mathfrak D}(\alpha, \beta, 0, \nu)$ is {\bf not} unique.

\item[(3)]
In the expansion in Eq.(\ref{300}) the  ${\mathfrak D}(\alpha, \beta, 0, \nu)$ are elements of a group (the $HWP(d)/{\mathbb Z}(d)$ group), 
and therefore if we multiply two operators $\Theta_1\Theta_2$ we get an expansion of the same type, with
\begin{eqnarray}\label{400}
{\mathfrak W}[\Theta _1\Theta _2; (-1)^\nu \alpha ,(-1)^\nu \beta ,\nu)&=&\frac{1}{2d}\sum_{\nu_2}\sum _{\alpha_2, \beta_2}{\mathfrak W}[\Theta_1; (-1)^{\nu-\nu _2} (\alpha-\alpha_2) ,(-1)^{\nu-\nu _2} (\beta-\beta_2),\nu-\nu_2)\nonumber\\
&\times& {\mathfrak W}[\Theta_2; (-1)^{\nu _2} \alpha_2 ,(-1)^{\nu_2} \beta_2 ,\nu_2) \omega[2^{-1}(\alpha\beta_2-\alpha_2\beta)].
\end{eqnarray}

\item[(4)]
The Wigner-Weyl function obeys the marginal relations:
\begin{eqnarray}\label{121a}
&&\frac{1}{d}\sum_{\beta}{\mathfrak W}(\Theta; \alpha, \beta,\nu)=\bra{P;2^{-1}\alpha}\Theta {\mathfrak P}^{\nu+1}\ket{P;2^{-1}\alpha};\;\;\;\nu\in{\mathbb Z}(2)\nonumber\\
&&\frac{1}{d}\sum_{\alpha}{\mathfrak W}(\Theta; \alpha, \beta,\nu)=\bra{X;2^{-1}\beta}\Theta {\mathfrak P}^{\nu+1}\ket{X;2^{-1}\beta}\nonumber\\
&&\frac{1}{d}\sum_{\alpha, \beta}{\mathfrak W}(\Theta; \alpha, \beta,\nu)={\rm Tr}(\Theta {\mathfrak P}^{\nu+1}).
\end{eqnarray}
This combines Eqs(\ref{AX22}),(\ref{AX11}).
In the case that $\Theta {\mathfrak P}^{\nu+1}$ is a density matrix (e.g., when $\nu=1$ and $\Theta$ is a density matrix) the right hand sides are probabilities.

Additional marginal relations are:
\begin{eqnarray}\label{121ab}
&&\frac{1}{2d}\sum_{\alpha, \beta, \nu}{\mathfrak W}(\Theta; \alpha, \beta,\nu)={\rm Tr}(\Theta {\varpi}_0)\nonumber\\
&&\frac{1}{2d}\sum_{\alpha, \beta, \nu}(-1)^{\nu+1}{\mathfrak W}(\Theta; \alpha, \beta,\nu)={\rm Tr}(\Theta {\varpi}_1).
\end{eqnarray}

\end{itemize}
\end{proposition}
\begin{proof}
\mbox{}
\begin{itemize}
\item[(1)]
This follows from Eq.(\ref{WW2}).
\item[(2)]
We prove Eq.(\ref{300}) if we multiply both sides of Eq.(\ref{350}) with $\Theta_{k\ell}$, and then sum over $k,\ell$.
We get a more general relation using Eq.(\ref{DFF2}):
\begin{eqnarray}
\Theta=\frac{1}{d}\sum_{\nu=0}^1\left[\lambda_\nu \sum _{\alpha ,\beta} {\mathfrak W}[\Theta; (-1)^{\nu+1} \alpha ,(-1)^{\nu+1} \beta ,\nu) {\mathfrak D}(\alpha, \beta, 0, \nu)\right ]
;\;\;\;\lambda_0=\lambda;\;\;\;\lambda_1=1-\lambda
\end{eqnarray}
Eq.(\ref{300}) is the special case with $\lambda=\frac{1}{2}$.
It is seen that  due to the redundancy, the expansion of $\Theta$ in terms of ${\mathfrak D}(\alpha, \beta, 0, \nu)$, is not unique.

\item[(3)]
Multiplication of two operators $\Theta_1, \Theta_2$ gives
\begin{eqnarray}
\Theta_1\Theta_2&=&\frac{1}{(2d)^2}\sum_{\nu_1,\nu_2}\sum _{\alpha _1,\beta_1, \alpha _2,\beta_2 } 
{\mathfrak W}[\Theta_1; (-1)^{\nu _1} \alpha_1 ,(-1)^{\nu_1} \beta_1 ,\nu_1) {\mathfrak W}[\Theta_2; (-1)^{\nu _2} \alpha_2 ,(-1)^{\nu_2} \beta_2 ,\nu_2) \nonumber\\
&\times& \omega[2^{-1}(\alpha_1\beta_2-\alpha_2\beta_1)]{\mathfrak D}(\alpha_1+\alpha_2, \beta_1+\beta_2, 0, \nu_1+\nu_2).
\end{eqnarray}
We change the variables $\alpha_1, \alpha_2$ into $\alpha=\alpha_1+\alpha_2, \alpha_2$ (and similarly for $\beta$), and also $\nu_1, \nu_2$ into $\nu=\nu_1+\nu_2, \nu_2$, and we get Eq.(\ref{400}).

\item[(4)]
This follows  from Eqs.(\ref{121}), (\ref{70}).

\end{itemize}
\end{proof}

\begin{example}
In $H(d)$ let  $\rho=\ket{X;j}\bra{X;j}$. In this case we get
\begin{eqnarray}
{\mathfrak W}(\rho; \alpha, \beta, \nu)=\omega[2^{-1}\alpha\beta +\alpha j(-1)^\nu]\delta(\beta,2\nu j).
\end{eqnarray}
This is a generalisation of example \ref{ex1}.
Then the expansion in Eq.(\ref{300}) is
\begin{eqnarray}
\ket{X;j}\bra{X;j}=\frac{1}{2d}\sum_\alpha \left[\omega(\alpha j){\mathfrak D}(\alpha, 0,0,0)+{\mathfrak D}(\alpha, 2j,0,1)\right ].
\end{eqnarray}

\end{example}
\begin{example}\label{ex7}
In $H(3)$ let  $\Theta=\ket{f}\bra{f}$ where $\ket{f}$ is given in Eq.(\ref{state}). Results for the unified Wigner-Weyl function are given in table \ref{t1}.
We can check that the last constraint in Eq.(\ref{121a}) holds.
\end{example}

\section{Discussion and further work}

Displacement operators in the Heisenberg-Weyl group, are used as quantum gates in quantum circuits (e.g., \cite{NC}).
More generally phase space methods, coherent states and the Wigner function play an important role  in quantum information.

In the present paper, the Heisenberg-Weyl group $HW(d)$ has been extended into the Heisenberg-Weyl-parity group $HWP(d)$ that incorporates parity transformations.
The $HWP(d)$ group contains all the elements of $HW(d)$ group (which is a subgroup of $HWP(d)$), plus new elements which are Fourier transforms of the former ones (Eq.(\ref{WW2})).

The first extension was the dihedral group that combines parity with displacements in the  momentum direction only (or in the position direction only).  
$HWP(d)$ is a generalised version of the dihedral group.
 
 The solvability of the $HW(d)$ group, of the dihedral group $\Delta_d(Z)$, and of the $HWP(d)$ group, has been discussed. In the present  physical context,
 it uses commutators that displace quantum states along loops in the discrete phase space. 
 In the case of $HWP(d)$ the solvability class is $3$, and it uses the commutators ${\cal L}_3^{(1)}$ and the `commutators of commutators' ${\cal L}_3^{(2)}$. 
 The latter goes around one loop, then around a second loop, then around 
the first loop in the opposite direction, and then around the second loop in the opposite direction.

 Two sets of coherent states related $HW(d)$ and $HWP(d)$ groups, have been studied. Related concepts like the Bargmann coefficients and the Husimi $Q$-function have also been discussed.
The  $2d^2$ coherent states related to the $HWP(d)$ group consist of $d^2$ coherent states related to its $HW(d)$ subgroup, plus $d^2$ extra coherent states which are 
related through a Fourier transform with the former ones (Eq.(\ref{WW22})). 

Both sets of coherent states have been used for expansions of an arbitrary state, with Bargmann coefficients (in section 6). 
Examples have shown that in noisy cases expansion in terms of the $2d^2$ coherent states, is advantageous (smaller error) in comparison to expansion in terms of the $d^2$ coherent states. 

The $HWP(d)$ group leads to a natural unification of the Wigner and Weyl functions.  The properties of the unified Wigner-Weyl function have been discussed in proposition \ref{pro7}.
A numerical example of the Wigner-Weyl function is given in table \ref{t1}.

We next summarise some directions for further work:
\begin{itemize}

\item

Central group extensions and group cohomology have been used for $HW(d)$ (e.g., \cite{VOUR}) and this approach could be applied to $HWP(d)$.
\item
In the present paper we considered Hilbert spaces with variables in ${\mathbb Z}(d)$ with odd dimension. Further work is need for extension of these ideas to the case of even dimension, and for 
unified approaches\cite{WQ1,WQ2,WQ3}.

\item
In the present paper we considered variables in ${\mathbb Z}(d)$ (with odd $d$). 
There has been a lot of work on `Galois quantum systems' with variables in a Galois field $GF(p^n)$ where $p$ is a prime number \cite{GG1,GG2,GG3,GG4}.
The present approach that incorporates the parity operator into the Heisenberg-Weyl group with a semidirect product, and also studies the enlarged group $HWP(d)$ 
as a solvable group, could be studied in this context also.
\item
In the present paper we studied the $HWP(d)$ group whose solvability class is $3$.
It will be interesting to study other groups of larger solvability class $n$, and the corresponding commutators ${\cal L}^{(1)}$,..., ${\cal L}^{(n-1)}$.
\item

Further work can study SIC-POVM, in the context of the $2d^2$ coherent states, introduced in this paper (Eq.(\ref{SIC})).
This can bring a new angle to the SIC-POVM problem.

\end{itemize}
The work is a contribution to phase space methods for systems with finite-dimensional Hilbert space, which play an important role within the area of quantum information.

\end{document}